\documentclass[11pt,letterpaper]{article}
\usepackage[margin=.9in]{geometry}

\usepackage{graphicx}
\usepackage{epstopdf}
\usepackage{graphicx}
\usepackage{xcolor}
\usepackage[colorlinks=true,citecolor=blue,linkcolor=magenta]{hyperref}
\usepackage{mathtools}
\usepackage{amsfonts}
\usepackage{amsmath}
\usepackage{amssymb}
\usepackage{amsfonts}
\usepackage{amsthm}
\usepackage{tikz}
\usepackage{comment}
\usepackage{listings}
\usepackage{cleveref}
\usepackage{float}
\usepackage{authblk}

\newcommand{\h}{\mathcal{H}}
\newcommand{\comm}{{\textnormal{comm}}}
\newcommand{\unif}{{\textnormal{unif}}}

\allowdisplaybreaks

\DeclareMathOperator{\Span}{span}

\DeclareMathOperator{\Id}{Id}
\DeclareMathOperator{\SL}{SL}

\newtheorem{theorem}{Theorem}
\newtheorem{lemma}[theorem]{Lemma}
\newtheorem{corollary}[theorem]{Corollary}
\newtheorem{proposition}[theorem]{Proposition}

\newtheorem{problem}{Problem}

\newcommand{\End}{{\mathrm{End}}}

\newcommand{\msph}{\mathcal{M}_{\mathrm{sph}}}

\newcommand{\m}{\mathcal{M}}
\newcommand{\x}{\mathcal{X}}
\newcommand{\y}{\mathcal{Y}}
\newcommand{\z}{\mathcal{Z}}

\newcommand{\w}{\mathcal{W}}

\usepackage[backend=bibtex8,sorting=none,citestyle=numeric-comp,firstinits=true,doi=false,isbn=false,url=false,maxbibnames=6]{biblatex}
\renewbibmacro{in:}{}
\title{Randomness of exact unitary designs under symmetry}
\author[1]{Christopher Vairogs}
\author[2,3]{Felix Leditzky}
\affil[1]{Department of Physics, University of Illinois Urbana-Champaign, 1110 W Green St Loomis Laboratory, Urbana, IL 61801, USA.\thanks{\href{mailto:christopher.vairogs@gmail.com}{christopher.vairogs@gmail.com}}}
\affil[2]{Department of Mathematics, University of Illinois Urbana-Champaign, 1409 W Green St,
Urbana, IL 61801, USA.\thanks{\href{mailto:leditzky@illinois.edu}{leditzky@illinois.edu}}}
\affil[3]{Illinois Quantum Information Science and Technology Center (IQUIST), University of Illinois Urbana-Champaign, 1101 W Springfield Ave, Urbana, IL 61801, USA}
\date{}

\begin{document}

\maketitle

\begin{abstract}
      We define the unitary design strength of a finite unitary ensemble to be the maximal integer $t$ for which its $t$-th statistical moment is identical to that of the unitary group. In the presence of physical symmetry, not all operators of an ensemble are symmetry-compatible. This motivates us to define the symmetric design strength of a finite unitary ensemble to be the maximal integer $t$ for which the $t$-th moments of its symmetry-compatible sub-ensemble match those of the group of symmetry-compatible unitaries. We show that in sufficiently high dimension and for a broad family of symmetries, which includes the global on-site $\mathrm{U}(1)$ and $\mathrm{SU}(d)$ symmetries, there exist finite unitary ensembles of arbitrarily high design strength whose symmetric design strength is strictly higher. Ensembles with this property may be more versatile under symmetry constraints for protocols that rely on unitary randomization than they are without symmetry. Our results extend the previous work [Mitsuhashi and Yoshioka, PRX Quantum 4.4 (Nov.~2023)] showing that the symmetric design strength of the Clifford group is strictly upper bounded by its unitary design strength for symmetries that do not trivialize the symmetry-compatible unitaries. Due to the special significance of unitary designs with a group structure, we also derive a variety of bounds on the unitary and symmetric design strengths of finite unitary ensembles whose operators form a group. In particular, we provide a simple proof that an arbitrary ensemble of operators forming a finite group must have a symmetric design strength of at most two under the global on-site $\mathrm{U}(1)$ and $\mathrm{SU}(d)$ symmetries. Along the way, we show that there exist uniformly weighted groups of arbitrarily high symmetric design strength in arbitrary dimension and analyze their structure. 
\end{abstract}

\section{Introduction}
Randomness is a unifying notion throughout modern quantum information science. It serves not only as a key theoretical concept needed to study many fundamental problems throughout physics, featuring prominently in the topic areas of information scrambling~\cite{Landsman2019, Mi2021} and quantum chaos~\cite{Roberts2017, Leone2021}, but also as a resource for numerous quantum information processing protocols, such as decoupling~\cite{horodecki2007negative,hayden2008decoupling,dupuis2014decoupling}, shadow tomography~\cite{Elben2022, Huang2020} and randomized benchmarking~\cite{Wallman2014, Kliesch2021}. The essential tool used to study randomness in these contexts is that of a unitary $t$-design, which is an ensemble of unitary operators whose $t$-th statistical moments match those of the uniformly random Haar measure on the unitary group. In this way, the maximal $t$ for which an ensemble of operators is a unitary $t$-design is a reflection of their randomness, with higher values of $t$ corresponding to a greater degree of randomness.

Simultaneously, symmetry is a ubiquitous concept throughout all of physics. Quantum systems in realistic settings may be subject to constraints imposed by various symmetries. Consequently, unitary evolutions produced by some random process must respect the relevant symmetry constraints~\cite{Hulse2024, Marvian2022}. Since arbitrary unitary operators sampled from the full Haar measure may not be realizable under the symmetry constraints, a mathematical construct generalizing the ordinary notion of a unitary design is needed to capture the notion of randomness. To this end, a \textit{symmetric $t$-design} has been introduced~\cite{Mitsuhashi_2023, liu2024unitarydesignsrandomsymmetric, hearth2025, li2024_1, li2024_2} as an ensemble of operators that commutes with the relevant symmetry group and whose $t$-th statistical moments match those of the Haar measure over the group of symmetry-respecting unitary operators, i.e., those operators that commute with the symmetry group. Symmetric $t$-designs may retain the utility of unitary designs for many quantum information processing applications adapted to the symmetry-constrained setting. For instance, classical shadow tomography of symmetry-respecting observables is made viable through sampling from a symmetric 3-design, as it is for general observables via sampling from unitary 3-designs~\cite{sauvage2024, Hearth2024}. 

    It is of natural interest to understand how imposing a symmetry affects the degree of randomness of a unitary $t$-design. Given an ensemble of unitary operators, it is not immediately clear how the the \textit{unitary design strength}, defined as the maximal $t$ for which the ensemble forms a $t$-design, is related to its \textit{symmetric design strength}, defined as the maximal $t'$ for which the sub-ensemble consisting of those unitaries that are realizable under the relevant symmetry forms a symmetric $t'$-design. If one can find an ensemble whose symmetric design strength exceeds its unitary design strength, this would not only suggest an intriguing relationship between symmetry and randomness, but would also show that one can overcome the limitations of symmetry constraints in the context of randomized quantum protocols by simply sampling from the symmetric unitaries. 
    
    This problem was previously studied by Mitsuhashi and Yoshioka~\cite{Mitsuhashi_2023} in the context of the uniformly weighted multi-qubit Clifford group, which is well-known to be at most a 3-design but not a 4-design \cite{kueng2015qubit,webb2016clifford,zhu2017multiqubit,zhu2016}. Notably, \textcite{Mitsuhashi_2023} showed that there does not exist a symmetry group for which the subgroup of symmetry-respecting Clifford operators becomes a symmetric 4-design in a non-trivial way. Furthermore, for the physically-relevant global on-site $\mathrm{U}(1)$ and $\mathrm{SU}(q)$ symmetries, the symmetry-respecting Clifford operators only form a symmetric 1-design. Thus, it remains open whether there can exist a symmetry and an ensemble of unitary operators whose symmetric design strength can exceed its unitary design strength. 

    On a separate note, all unitary designs that are uniformly weighted finite groups (i.e., \textit{group designs}) have been fully characterized~\cite{Bannai, guralnick2005}. Such constructions are known to have a rich mathematical structure~\cite{Gross2007}, with applications in randomized benchmarking~\cite{Kliesch2021} and (thrifty) shadow estimation~\cite{Helsen2023}. It is also known that in Hilbert space dimension $d>2$, group $t$-designs for $t > 3$ cannot exist~\cite{Bannai,kaposi2026}. However, this does not preclude the possibility that the subset of unitary operators in a finite group design that are realizable under the relevant symmetry form a symmetric $t$-design for $t>3$, thereby breaking the $t= 4$ barrier. 

    In this work, we investigate how symmetry influences the degree of randomness of general finite unitary ensembles. 
    Our main results, which are discussed in more detail in \Cref{sec:main-results}, can be explained on a high level as follows:
    First, moving beyond the Clifford group paradigm of~\cite{Mitsuhashi_2023}, we prove the existence of unitary designs that become strictly more random as symmetric designs after restricting to the symmetry-respecting subset for a broad range of symmetries. In particular, we show that for arbitrary positive integers $t'\geq t$, there exists a finite unitary ensemble with unitary design strength of $t$ and a symmetric design strength of $t'$ for symmetries complying with an easily satisfiable representation-theoretic condition in all sufficiently high dimensions. Such symmetries include the global on-site $\mathrm{U}(1)$ and $\mathrm{SU}(q)$ symmetries, as well as subsystem permutation symmetry. 
    
    Next, we specialize to unitary ensembles that have the structure of a uniformly weighted finite group. We argue that in arbitrary dimension $d$ and for arbitrary $t> 0$, there are symmetry groups for which symmetric $t$-designs that are uniformly weighted finite groups exist, unlike in the asymmetric case. However, for $t> 5$, the elements of such groups are necessarily diagonal with respect to a basis determined by the symmetry. Thus, uniformly weighted groups that are
    symmetric $t$-designs for $t>5$ have a more elementary structure.
    Restricting to special cases of symmetry groups, we also extend a result of~\cite{Mitsuhashi_2023} by showing that for a global on-site $\mathrm{U}(1)$ and $\mathrm{SU}(q)$ symmetry, any finite unitary ensemble for which the operators form a group can be at most a symmetric $2$-design. We derive an analogous statement for symmetries described by finite groups as well. Finally, we show that for any $t>1$ and in arbitrary dimension, there exist a symmetry group and a uniformly weighted finite group $G$ of unitaries such that $G$ forms at most a unitary 1-design but its symmetry-respecting subset forms a symmetric $t$-design; moreover, for $t\leq 5$, there is no bound on the dimension for which the operators in this group may be taken to have a non-diagonal structure. Thus, the degree of randomness of a finite unitary design may be augmented under symmetry even when it simultaneously has a finite group structure.

\section{Definitions and Preliminaries}\label{sec:defs-and-prelims}
In all that follows, let $U(\mathcal{H})$ and $\End(\h)$ denote the unitary group and algebra of linear operators over a finite-dimensional Hilbert space $\mathcal{H}$. In our work, symmetries are described by subgroups $S\leq U(\h)$. For any subset $\mathcal{X}\subset U(\mathcal{H})$ and subgroup $S \leq U(\mathcal{H})$, define the set
\begin{equation}\label{eq:symmetrization-def}
    \mathcal{X}_S \coloneqq \{g \in \mathcal{X}: [g, s] = 0~\text{for all}~s\in S\};
\end{equation}
the set $\x_S$ is known as the \textit{symmetry-respecting} or \textit{symmetry-compatible} subset of $\x$ for symmetry group $S$. Given an arbitrary group $K$, \textit{finite} subset $\mathcal{X}\subset K$, and probability mass function (pmf)  $p\colon \x \to [0,1]$, we call the ordered pair $(\x, p)$ a \textit{weighted $K$-subset}. For any operator $x\in \End(\h)$, let $\overline{x}$ denote its complex conjugation with respect to the computational basis. If $K$ is a closed subgroup of $U(\h)$, we say that a weighted $K$-subset $(\x, p)$ is a $t$-\textit{design in $K$} if 
\begin{equation}\label{eq:moment-equality}
    \sum_{g\in \mathcal{X}} p(g) g^{\otimes t} \otimes \overline{g}^{\otimes t} = \int_{K} dk~k^{\otimes t}\otimes \overline{k}^{\otimes t},
\end{equation}
where the integration is with respect to the (normalized) Haar measure over $K$, which exists due to the compactness of $K$.
For such $(\x, p)$, we will sometimes simply refer to $(\x, p)$ as a \textit{unitary $t$-design} or simply \textit{$t$-design} when $K = U(\h)$. Note that if $(\x, p)$ is a $t$-design in  $K$, then it is a $t'$-design in $K$ for all positive integers $t'\leq t$. We will also refer to the LHS of~\eqref{eq:moment-equality} as the \textit{$t$-th moment} of $(\x, p)$. In our work, we will denote the uniform pmf over a set of operators by ``$\mathrm{unif}$", so that $(\x, \mathrm{unif})$ refers to a uniformly weighted $K$-subset when $\x \subset K$.

One may verify that for any subgroup $S\leq U(\h)$, the set $U(\h)_S$, whose definition is provided by~\eqref{eq:symmetrization-def}, is indeed a closed subgroup of $U(\h)$, so that the notion of a weighted $U(\h)_S$-subset $(\x, p)$ that is a $t$-design in $U(\h)_S$ is well-defined.
For such $(\x, p)$, we will sometimes refer to $(\x$, p) as an \textit{$S$-symmetric $t$-design under $p$} or simply $S$-\textit{symmetric $t$-design}. If a weighted $U(\h)$-subset $(\x, p)$ is a $t$-design but not a $(t+1)$-design in $U(\h)$, then we write 
\begin{equation}
    t_{\mathrm{max}}(\x, p) \coloneqq t.
\end{equation}
If a weighted $U(\h)$-subset $(\x, p)$ is such that $(\x_S, p_S)$ is a $t$-design but not a $(t+1)$-design in $U(\h)_S$, then we write
\begin{equation}
    t_{\mathrm{max}}^{(S)}(\x, p) \coloneqq t.
\end{equation}
We refer to $t_{\max}(\x, p)$ (resp. $t_{\max}^{(S)}(\x, p)$) as the \textit{unitary design strength} or simply \textit{design strength} (resp. \textit{$S$-symmetric design strength} or simply \textit{symmetric design strength}) of a weighted $U(\h)$-subset $(\x, p)$. Any weighted $U(\h)$-subset (resp.~$U(\h)_S$-subset), which is necessarily finite per our definition, will have a finite design strength (resp.~symmetric design strength), respectively~\cite{liu2024unitarydesignsrandomsymmetric}.

Given a finite subset $\mathcal{X} \subset U(\mathcal{H})$ with pmf $p\colon \x\to [0,1]$ and a subgroup $S \leq U(\mathcal{H})$ for which $\sum_{g \in \x_S} p(g) > 0$, we construct the pmf $p_S\colon \x_S \to [0,1]$ by defining
\begin{equation}\label{eq:ps-def}
    p_S(g) \equiv \frac{p(g)}{\sum_{g' \in \mathcal{X}_S} p(g')}.
\end{equation}
The probability $p_S(g)$ may be thought of as the conditional probability of sampling $g$ from the distribution $p$ given that $g$ lies within the subset $\mathcal{X}_S$.

Let $S$ be a closed subgroup of $U(\mathcal{H})$. In this case, elementary representation theory gives us powerful tools for working with $S$-symmetric designs. The space $\mathcal{H}$ carries natural representations $\xi$ and $\sigma$ of $S$ and $U(\h)_S$ given simply by $\xi(s) = s$ and $\sigma(u) = u$, respectively.
Under the representations $\xi$ and $\sigma$, the space $\mathcal{H}$ decomposes as
\begin{align}
    \mathcal{H} &\cong \bigoplus_{\lambda = 1}^n V_\lambda \otimes W_\lambda \label{eq:decomposition}\\
    \xi(s) &\cong\label{eq:s-action} \bigoplus_{\lambda = 1}^n \xi_\lambda(s) \otimes I_{W_\lambda}\\
    \sigma(u) &\cong\label{eq:u-h-s-action} \bigoplus_{\lambda = 1}^n I_{V_\lambda} \otimes \sigma_\lambda(u)
\end{align}
where the $V_\lambda$ carry the inequivalent irreducible representations $\xi_\lambda$ of $S$ that appear in $\mathcal{H}$ and the $W_\lambda$ are their corresponding multiplicity spaces; here, the $W_\lambda$ carry representations $\sigma_\lambda$ of $U(\h)_S$, though the $\sigma_\lambda$ need not be irreducible. We note that~\eqref{eq:u-h-s-action} follows from~\eqref{eq:decomposition}-\eqref{eq:s-action} and Schur's Lemma. 
By the isomorphisms of representations described in~\eqref{eq:decomposition}-\eqref{eq:u-h-s-action}, there exists an orthonormal basis 
        \begin{equation}
            \beta \coloneqq \{|\lambda, i, j\rangle: \lambda \in [n], i\in [\dim V_\lambda], j \in [\dim W_\lambda]\}
        \end{equation}
        of $\h$ such that for all $\lambda \in [n], j \in [\dim W_\lambda]$, the subspaces
        \begin{equation}
            V_{\lambda, j} \coloneqq \mathrm{span}\left\{ |\lambda, i, j\rangle: i \in [\dim V_\lambda]\right\} 
        \end{equation}
        are irreducible under $\xi$ and for all $\lambda \in [n], i \in [\dim V_\lambda]$, the subspaces
        \begin{equation}\label{eq:multiplicity-copies}
            W_{\lambda, i} \coloneqq \mathrm{span}\left\{ |\lambda, i, j\rangle: j \in [\dim W_\lambda] \right\}
        \end{equation}
        are invariant under $\sigma$. 
        Furthermore, $V_{\lambda, i} \cong V_\lambda$ and $W_{\lambda, j} \cong W_\lambda$ as $S$ and $U(\h)_S$-representations, respectively. 

        Another key tool we will use in our analysis is the \textit{commutant} of a set of operators. For any collection $\mathcal{X}$ of linear operators over a Hilbert space $\mathcal{H}$, we define the \textit{commutant} of $\x$ by
        \begin{equation}\label{eq:comm-def}
            \comm(\mathcal{X}) \coloneqq \{L \in \End(\mathcal{H}):  [L, O] = 0 \text{ for all } O \in \mathcal{X}\}.
        \end{equation}
        It may be checked using elementary arguments that if $G \leq H $ are compact subgroups of $U(\mathcal{H})$ for some Hilbert space $\mathcal{H}$, then the equality of Haar integrals
        \begin{equation}
            \int_H dh \ h^{\otimes t} \otimes \overline{h}^{\otimes t} = \int_G dg \ g^{\otimes t} \otimes \overline{g}^{\otimes t} \label{eq:integrals}
        \end{equation}
        holds if and only if~\cite{Gross2007}
        \begin{equation}\label{eq:comm-equality}
            \comm(\{h^{\otimes t}: h\in H\}) = \comm(\{g^{\otimes t}: g\in G\}).
        \end{equation}

\section{Main results}
\label{sec:main-results}

Since not all unitaries from a weighted $U(\h)$-susbet $(\x, p)$ may be implemented in a symmetric setting, a straightforward strategy for achieving an $S$-symmetric design with the same or greater design strength would be to simply implement those operators that are symmetry-compatible. To be more precise, the process of implementing the ensemble $(\x, p)$ may simply be thought of as sampling from a  collection of $|\x|$ labels encoding the unitaries in $\x$ that are distributed according to $p$ and performing the unitary evolution corresponding to the sampled label in a quantum circuit. On the other hand, implementing just the symmetry-compatible operators can be modeled by sampling a label from the collection of $|\x|$ labels, checking whether the corresponding unitary is symmetry-respecting, performing the unitary evolution if it is symmetry-respecting, discarding the sample if it is not symmetry-respecting, and repeating the procedure until the sampled label produces a symmetry-respecting unitary. Then the implemented unitaries from $\x_S$ will follow the conditional distribution $p_S$. However, it is not immediately clear whether the mathematical framework of $S$-symmetric designs even allows for the $S$-symmetric design strength of $(\x_S, p_S)$ to match or exceed the unitary design strength of $(\x, p)$. To this end, we pose our first problem:

\begin{problem}\label{prob-1}
Does there exist a symmetry group $S\leq U(\h)$ and a weighted $U(\h)$-subset $(\x, p)$ such that \begin{align}
    t_{\max}^{(S)}(\x, p) \geq t_{\max}(\x, p) > 0?
\end{align}
\end{problem}

While variations of this problem involving approximate designs and approximate symmetric designs (\textit{i.e.}, weighted subsets for which the moment equality in~\eqref{eq:moment-equality} holds only approximately) may be formulated, our work deals primarily with the exact case. Beyond general mathematical interest, we specialize to the exact case because properties arising in the exact case inform us about those arising in the approximate case. Furthermore, exact designs are preferable in protocols which demand enough repeated uses of a design for significant error to accumulate~\cite{nakata2021}. 

Problem~\ref{prob-1} should be viewed as a generalization of a problem studied in a prior work by Mitsuhashi and Yoshioka~\cite{Mitsuhashi_2023}. Recall that the $N$-qubit Clifford group $\mathcal{C}_N$ is defined as $\mathcal{C}_N\coloneqq \{U \in U((\mathbb{C}^2)^{\otimes N}): UPU^\dagger \in \mathcal{P}_N\}$, where $\mathcal{P}_N \coloneqq \{\pm 1, \pm i\}\cdot \{I, X, Y, Z\}^{\otimes N}$ is the group generated by Pauli matrices $I, X, Y, Z$ on each of the $N$ qubits. Mitsuhashi and Yoshioka generalize to the symmetric setting the well-known theorem that the $N$-qubit Clifford group satisfies $t_{\max}(\mathcal{C}_N, \mathrm{unif}) = 3$. In particular, they show that if $S\leq U((\mathbb{C}^2)^{\otimes N})$ is a symmetry group for which $U((\mathbb{C}^2)^{\otimes N})_S \neq \{e^{i\theta}I: \theta \in \mathbb{R}\}$, i.e., the symmetric unitaries are not trivialized, then  
\begin{equation}\label{eq:mitsuhashi-yoshioka-1}
    t_{\max}^{(S)}(\mathcal{C}_N, p) \leq t_{\max}^{(S)}(\mathcal{C}_N, \mathrm{unif}) \leq 3 = t_{\max}(\mathcal{C}_N, \mathrm{unif}), 
\end{equation}
for any pmf $p$ over $\mathcal{C}_n$ and $\unif$ denoting the uniform distribution. Moreover, in the practical cases of global on-site $U(1)$ and $SU(2)$ symmetries, it was shown that
\begin{equation}
    t_{\max}^{(S)}(\mathcal{C}_N, \mathrm{unif}) = 1.
\end{equation}
Thus, by implementing just the symmetry-compatible Clifford gates, we can at most hope to retain 3-designness, and are limited to 1-designness in the physically relevant examples of global on-site $U(1)$ and $SU(2)$ symmetries. However, this does not preclude the possibility of achieving a greater symmetric design strength by implementing the symmetry-comptatible operators of some other design. 

We answer Problem~\ref{prob-1} in the affirmative by demonstrating that for a broad family of symmetry groups, there exists a weighted $U(\h)$ subset whose symmetric design strength may exceed their unitary design strength by an arbitrary amount. This suggests that symmetry groups for which there exist ensembles of operators whose symmetric design strength exceeds or matches their unitary design strength are not aberrations, but rather generic. More precisely, we prove the following theorem.
\begin{theorem}\label{thm:main-thm-1}
    Let $t' \geq t> 0$ be arbitrary. There exists a $d_0>0$ such that for any Hilbert space $\h$ of dimension $\dim \h \geq d_0$, and for any closed subgroup $S \leq U(\h)$ for which the natural representation of $S$ on $\h$ has a multiplicity-free irreducible subrepresentation, there exists a weighted $U(\h)$-subset $(\x, p)$ satisfying
    \begin{align}
        t_{\max}^{(S)}(\x, p) \geq t' > t = t_{\max}(\x, p). 
    \end{align}
    In particular, if $1\leq t \leq 3$, then we may take $d_0 = 3$ in the preceding statement.
\end{theorem}

It should be noted that the restriction to symmetry groups possessing multiplicity-free irreps is largely a technical restriction and is imposed to ease calculations in the derivation of our result. Moreover, a variety of physically-relevant symmetry groups do satisfy this condition. In particular, we have the following corollary of Theorem~\ref{thm:main-thm-1}:
\begin{corollary}\label{cor:symmetry-boosted-designs}
        Let $N\geq 1$, $d\geq 2$. Let $\h_1 = (\mathbb{C}^d)^{\otimes N} = \h_3$, and $\h_2 = (\mathbb{C}^2)^{\otimes N}$. Define subgroups $S_i \leq U(\h_i), i = 1,2,3$ by 
        \begin{align}
            S_1 &\coloneqq \left\{u^{\otimes N}: u \in \mathrm{SU}(d)\right\} \cong \mathrm{SU}(d)  \\
            S_2 &\coloneqq \left\{(e^{i\theta Z})^{\otimes N}: \theta \in \mathbb{R}\right\} \cong \mathrm{U}(1)\\
            S_3 &\coloneqq \left\{P_d(\pi): \pi \in S_N\right\} \cong S_N, 
        \end{align}
        where $Z$ denotes the Pauli $Z$ operator on $\mathbb{C}^2$ and $P_d(\pi)$ denotes the qudit permutation operator associated with $\pi\in S_N$. 
        Let $t' \geq t \geq 1$ be arbitrary integers. For $i = 1,2$ and all $N$  sufficiently large, there exists a weighted $U(\h_i)$-subset $(\x^{(i)}, p^{(i)})$ such that 
        \begin{equation}\label{eq:t-max-ineqs-2-main}
            t_{\max}^{(S_i)}(\x^{(i)}_{S_i}, p^{(i)}_{S_i}) \geq t' \geq t = t_{\max}(\x^{(i)}, p^{(i)}).  
        \end{equation}
        In particular, for $i = 1,2$, $N\geq 2$, and $1\leq t \leq 3$, there exists a weighted $U(\h_i)$-subset $(\x, p)$ for which~\eqref{eq:t-max-ineqs-2-main} holds. For $i= 3$ and all $N$ that are sufficiently large multiples of $d$, there exists a weighted $U(\h_i)$-subset $(\x, p)$ such that~\eqref{eq:t-max-ineqs-2-main} holds. In particular, this inequality is achieved if $1\leq t\leq 3$ and $N$ is any positive multiple of $d$. 
\end{corollary}
We prove Theorem~\ref{thm:main-thm-1} and Corollary~\ref{cor:symmetry-boosted-designs} in Section~\ref{sec:symmetry-boosted-design-proofs}.

It should be noted that in general, it is not difficult to find an ensemble whose symmetric strength exceeds its unitary design strength. One may check that for symmetry group $S$ such that $U(\h)_S\neq U(\h)$, any $t$-design in $U(\h)_S$ will not even be a 1-design in $U(\h)$. However, this fact does not at all establish whether there exists a general unitary $t$-design whose strength is boosted under symmetry for even low $t>0$. If we wish to generalize the Mitsuhashi and Yoshioka result~\eqref{eq:mitsuhashi-yoshioka-1} beyond the Clifford group, such a case must be handled. Naively, one may try to construct a symmetry-boosted unitary design by taking any $U(\h)$-subset $(\y, q)$ for which $t_{\max}(\y, q) = t >0$, exhibiting a construction of an $S$-symmetric $t'$-design $(\z, r)$ for some $t'>t$, and appending the symmetric design to the unitary design by defining $(\x, p) \coloneqq  (\y \cup \z, (1-\varepsilon)q + \varepsilon r)$ for some $\varepsilon>0$. If one can show that we can find a $\y$ as above for which we have $\y\cap \z = \varnothing$, then $(\x_S, p_S) = (\z, r)$. In that case, choosing $\varepsilon$ to be sufficiently small will yield moments of $(\x, p)$ that approximate those of $(\y, q)$. Then $(\y, q)$ will be an approximate $t$-design, while it seems implausible that $(\x, p)$ would be an approximate $(t+1)$-design. However, such an argument, even after verification of its technical details, can at best only establish the existence of a symmetry-boosted \textit{approximate} design. In this way, the construction we employ in the proof of Theorem~\ref{thm:main-thm-1} plays a non-trivial role.

At a high level, we circumvent the aforementioned difficulties and prove Theorem~\ref{thm:main-thm-1} as follows. First, we exhibit a construction of an $S$-symmetric $t'$-design $(\y, q)$ and a weighted $U(\h)$-subset $(\z, r)$. Then we define $(\x, p) = (\y\z, q\star r)$, where $\star$ denotes a convolution of functions over $U(\h)$. With some elementary algebra, one may show that $(\x_S, p_S) = (\z, r)$, so that $t_{\max}^{(S)}(\x, p)\geq t'$. The left-invariance of the Haar measure then implies that $(\x, p)$ is a $t$-design if $(\z, r)$ is a $t$-design. However, demanding that $t_{\max}(\z, r) \leq t$ alone is not sufficient to ensure that $t_{\max}(\x, p) \leq t$, as may be seen from the constructions for unitary designs considered in~\cite{kaposi2026}. To avoid this issue, the majority of the proof is devoted to building $(\y, q)$ using convolutions of $(\z, q)$ with particular zeros of certain zonal spherical functions to ensure that $(\y, q)$ is a $t$-design that indeed yields $t_{\max}(\x, p) = t$. In doing so, we make use of the ideas presented by \textcite{bannai2020} in their construction of explicit unitary $t$-designs.

While our proof is constructive, it should also be noted that the unitary ensembles we construct are not necessarily practical. Nevertheless, our result concretely demonstrates that the mathematical structure of symmetric designs allows for the possibility of symmetry-boosted unitary designs for familiar symmetries. Identifying practical unitary ensembles whose symmetric design strength may exceed their unitary design strength is a topic for further research.

On a separate note, it is worth considering how symmetric design strength behaves for unitary ensembles with a group structure. It is known that if $\dim \h>2$, and $G$ is a finite subgroup of $U(\h)$, then 
\begin{equation}\label{eq:group-design-barrier}
    t_{\max}(G, p) \leq 3
\end{equation}
for any pmf $p$ over $G$~\cite{bannai2020}. In this way, a group structure throttles the randomness of designs. To understand how group structure limits symmetric design strength, we pose the following problem:

\begin{problem}
Are there universal bounds on $t_{\max}^{(S)}(G, \mathrm{unif})$ for finite subgroups $G$ of $U(\h)$ and if so, what are they? 
\end{problem}

We address this problem in the following theorem.

\begin{theorem}\label{thm:main-thm-2}
    The following hold:
    \begin{enumerate}
        \item If $S$ is a closed subgroup of $U(\h)$ whose natural representation on $\h$ is multiplicity-free, then for any $t>0$, there exists a finite group $H \leq U(\h)_S$ such that $t_{\max}^{(S)}(H, \mathrm{unif}) \geq t$. Consequently, for any finite-dimensional Hilbert space $\h$ and any $t>0$, there exists a symmetry group $S\leq U(\h)$ and finite group $H\leq U(\h)_{S}$ for which $t_{\max}^{(S)}(H, \mathrm{unif}) \geq t$.

        \item Let $G\leq U(\h)$ be finite and let $S\leq U(\h)$ be closed. If $t_{\max}^{(S)}(G, p) > 5$ for some pmf $p$ over $G$, then every element of $U(\h)_S$ is diagonal in the irrep basis of $S$ and, hence, $U(\h)_S$ is abelian.

        \item For any finite-dimensional Hilbert space $\h$, there exists a symmetry group $S\leq U(\h)$ and finite group $H\leq U(\h)_S$ for which $t_{\max}^{(S)}(H, \mathrm{unif}) = 5$ and $U(\h)_S$ is non-abelian.
    \end{enumerate}
\end{theorem}

Theorem~\ref{thm:main-thm-2} shows that in arbitrary dimension, there are no bounds per se on $t_{\max}^{(S)}(G, \mathrm{unif})$ for finite unitary subgroups $G$. However, finite unitary subgroups possessing symmetric design strength beyond $t=5$ do not have a particularly interesting structure, since they can only consist of diagonal unitaries. Nevertheless, at a symmetric design strength of $t=5$ in dimension $d>2$, a finite subgroup of $U(\h)_S$ may possess a non-diagonal structure, unlike in the asymmetric case. It is also worth considering symmetric design strength of finite groups in a few notable special cases:

\begin{theorem}\label{thm:main-thm-3}
    The following hold:
    \begin{enumerate}
        \item Let $G\leq U(\h)$ be finite and let $S\leq U(\h)$ be closed. Let $q\geq 2, N\geq 4$. If $\h = (\mathbb{C}^2)^{\otimes N}$ or $\h = (\mathbb{C}^q)^{\otimes N}$ and 
    \begin{align}
        S &= \left\{(e^{i\theta Z})^{\otimes N}: \theta \in \mathbb{R}\right\} \cong \mathrm{U}(1)
    \intertext{or}
        S &= \left\{u^{\otimes N}: u \in \mathrm{SU}(q)\right\} \cong \mathrm{SU}(q),
    \end{align}
    then $t_{\max}^{(S)}(G, p) \leq 2$ for any pmf $p\colon  G\to [0,1]$.

    \item  Let $G\leq U(\h)$ be finite and assume that $S\leq U(\h)$ is such that $|S| < \dim \h/2$. Then $t_{\max}^{(S)}(G, p) \leq 3$ for any pmf $p\colon  G\to [0,1]$.

    \item Let $G\leq U(\h)_S$ be finite and let $\SL(2, 5)$ denote the multiplicative group of $2\times 2$ determinant-one matrices with elements from the finite field of order 5. If $4\leq t_{\max}^{(S)}(G, p) \leq 5$ for some pmf $p$ over $G$, then 
        \begin{equation}
            G \leq \left\{ \bigoplus_{\lambda = 1}^n I_{V_\lambda} \otimes s_\lambda: s_\lambda \in \left\langle \exp\left(\frac{2\pi i}{k_\lambda} \right) I_{W_\lambda} \right\rangle \rho_\lambda(\SL(2, 5)) \right\}
        \end{equation}
        for some integers $k_1,\dots, k_n$ and irreducible representations $\rho_\lambda: \SL(2, 5) \to U(W_\lambda), \lambda \in [n]$ of either degree one or degree two; furthermore, for every $s_\lambda \in \langle \exp(2\pi i/k_\lambda) I_{W_\lambda}\rangle \rho_\lambda(\SL(2, 5))$, there exists an element $g\in G$ such that $\sigma_\lambda(g) = s_\lambda$.
    \end{enumerate}
\end{theorem}

Complete proofs of Theorems~\ref{thm:main-thm-2} and~\ref{thm:main-thm-3} may be found in Section~\ref{sec:group-thms}. The proof method of these theorems proceeds by noting that given a finite subgroup $G$, the symmetry-respecting subset $G_S$ will also be a finite group. Since $S$ is closed, the Hilbert space $\h$ decomposes into irreducible representations of $S$, each of which occurs with a certain multiplicity. Since $G$ commutes with $S$, Schur's Lemma implies that $G$ can act non-trivially only on the multiplicity spaces, and since $G_S$ is a group, its components over the multiplicity spaces are also groups. We then argue that if $G_S$ is a design, its components over the multiplicity spaces are also designs, and hence, subject to the restrictions imposed by the classification of finite group designs~\cite{bannai2020, guralnick2005}. In this way, we obtain our restrictions on the symmetric design strengths of finite groups.
    
The first part of Theorem~\ref{thm:main-thm-3} may be regarded as an extension of a result of Mitsuhashi and Yoshioka~\cite{Mitsuhashi_2023}, which shows that $t_{\max}^{(S)}(\mathcal{C}_N, \mathrm{unif}) \leq 1$ for any symmetry group of the form $S = \{u^{\otimes N}: u\in H\}$, where $H$ is a Lie subgroup of $U(\mathbb{C}^2)$. Our bound is weaker by one moment than the Mitsuhashi and Yoshioka bound and is specialized to the case of $\mathrm{U}(1)$ and $\mathrm{SU}(q)$ symmetries, but applies more generally to all finite unitary subgroups beyond the multi-qubit Clifford group. We find it interesting that this extension may be obtained from the short and simple argument outlined above relying on the classification of finite group designs~\cite{bannai2020,guralnick2005}, whereas the bound on the symmetric design strength of the Clifford group from~\cite{Mitsuhashi_2023} was obtained through a long sequence of detailed Pauli matrix calculations. The second part of Theorem~\ref{thm:main-thm-3} essentially says that if the symmetry group is finite and has an order that is small relative to $\dim \h$, then the $t=3$ barrier~\eqref{eq:group-design-barrier} on unitary design strength for finite groups in $\dim \h>2$ cannot be circumvented. Finally, the third part of Theorem~\ref{thm:main-thm-3} suggests that a finite subgroup of $U(\h)_S$ with a symmetric design strength of either $t = 4$ or $t=5$ is essentially built out of copies of the group $\SL(2, 5)$ and relative phases, and therefore has a highly restricted structure.

Finally, we revisit Problem~\ref{prob-1} in the context of unitary designs that simultaneously have a group structure. We note that such designs play a key role in randomized benchmarking protocols~\cite{Wallman2014} and the construction of explicit unitary designs~\cite{Gross2007, Roy_2009, Bannai2019, kaposi2026}. We pose the following problem:

\begin{problem} 
Does there exist a symmetry group $S\leq U(\h)$ and a finite subgroup $G \leq U(\h)$ such that
\begin{equation}
    t_{\max}^{(S)}(G, \mathrm{unif}) \geq t_{\max}(G, \mathrm{unif}) > 0?
\end{equation}
\end{problem}

We again answer affirmatively:
\begin{theorem}\label{thm:main-thm-4}
    Let $t>1$. Then for any finite-dimensional Hilbert space $\h$, there exists a finite group $G\leq U(\h)$ and symmetry group $S\leq U(\h)$ such that $U(\h)_S \neq \{e^{i\theta} I: \theta \in \mathbb{R} \}$ and \begin{equation}\label{eq:group-t-max-statement}
        t_{\max}^{(S)}(G, \mathrm{unif})\geq t > t_{\max}(G, \mathrm{unif}) = 1.
    \end{equation}
    Moreover, if $\dim \h$ is even and $t\leq 5$, then there exists a finite group $G\leq U(\h)$ and symmetry group $S\leq U(\h)$ such that~\eqref{eq:group-t-max-statement} holds and $G_S$ is non-abelian.
\end{theorem}
We prove Theorem~\ref{thm:main-thm-4} in Section~\ref{sec:proof-of-grp-des}. While this result is not as strong as Theorem~\ref{thm:main-thm-1} in that it only applies to unitary 1-designs, it concretely demonstrates that there is a family of group designs in unboundedly high dimension whose strength is boosted by some symmetry.

\section{Proofs of Theorem~\ref{thm:main-thm-1} and Corollary~\ref{cor:symmetry-boosted-designs}}\label{sec:symmetry-boosted-design-proofs}

\subsection{Setup and construction for symmetric designs}

We first establish a basic construction of $S$-symmetric $t$-designs for any $t>0$. While constructions from local random gates were considered in~\cite{hearth2025,liu2024unitarydesignsrandomsymmetric}, they had infinite cardinality, whereas we focus on finite ensembles in this paper. Let $\h$ be a $d$-dimensional Hilbert space and let $S\leq U(\h)$ be closed. Assume the notation of Section~\ref{sec:defs-and-prelims}, in particular the decompositions~\eqref{eq:decomposition}-\eqref{eq:u-h-s-action}.
\begin{proposition}\label{prop:symmetric-design-construction}
    Let $C_\lambda \coloneqq \langle \exp(2\pi i/(t+1))I_{W_\lambda}\rangle$, let $\mu_\lambda$ denote the uniform distribution over $C_\lambda$, and let $(H_\lambda, \nu_\lambda)$ be a $t$-design in $U(W_\lambda)$. Define
    \begin{equation}\label{eq:sym-design-construction}
        H \coloneqq \left\{ \bigoplus_{\lambda = 1}^n I_\lambda \otimes g_\lambda: g_\lambda \in C_\lambda H_\lambda \right\}
    \end{equation}
    and construct a probability distribution $p$ over $H$ by defining
    \begin{equation}
        p\left(\bigoplus_{\lambda = 1}^n I_{V_\lambda} \otimes g_\lambda \right) = \prod_{\lambda = 1}^n (\mu_\lambda \star \nu_\lambda)(g_\lambda).
    \end{equation}
    Then the weighted $U(\h)$-subset $(H, p)$ is a $t$-design in $U(\h)_S$. 
\end{proposition}
The idea behind Proposition~\ref{prop:symmetric-design-construction} is that by combining mixtures of operators from unitary $t$-designs in groups $U(W_\lambda)$ with certain relative phases over the multiplicity spaces, we may form $t$-designs in $U(\h)$. It is crucial that the relative phases appear in~\eqref{eq:sym-design-construction} via $C_\lambda$, for $(H, p)$ will not necessarily be a $t$-design in $U(\h)_S$ without them. Note that we could also form $t$-designs in $U(\h)_S$ by considering mixtures of operators from so-called \textit{strong} unitary $t$-designs in $U(W_\lambda)$, as in~\cite{bannai2020}, though we opt to consider a construction built out of the the more familiar notion of a unitary $t$-designs. The proof of Proposition~\ref{prop:symmetric-design-construction} follows from Proposition~\ref{prop:block-diag-designs} in Appendix~\ref{app:block-diag-designs}, which itself amounts to direct computation.

    \subsection{Overview of zonal spherical functions}

    Let us now fix a positive integer $t$. The space $\h^{\otimes t} \otimes \h^{\otimes t}$ carries a representation $\pi_t$ of $U(\h)$ given by
    \begin{equation}\label{eq:tensor-power-rep}
        \pi_t(u) = u^{\otimes t} \otimes \overline{u}^{\otimes t}
    \end{equation}    
    for $u\in U(\h)$. Let $\m(d)$ denote the collection of all $d$-tuples $(\mu_1, \dots, \mu_d)$ of integers such that $\mu_1 \geq \dots \geq \mu_d$.  The isomorphism classes of the finite-dimensional unitary irreps of $U(\h)$ are in one-to-one correspondence with elements of $\mathcal{M}(d)$, which specify their highest weights. In other words, we may associate a unitary $U(\h)$-irrep $(V_\mu, \zeta_\mu)$ with each $\mu \in \mathcal{M}(d)$ such that every $U(\h)$-irrep is isomorphic to some $(V_\mu, \zeta_\mu)$, and two irreps $(V_\mu, \zeta_\mu)$ and $(V_{\mu'}, \zeta_{\mu'})$ are isomorphic if and only if $\mu = \mu'$. Let $\mathcal{M}(d, t)$ denote the set of all $\mu \in \mathcal{M}$ such that
    \begin{equation}
        \sum_{\stackrel{i \in [d]:}{\mu_i > 0 }} \mu_i = \sum_{\stackrel{i \in [d]:}{\mu_i < 0}} |\mu_i| \leq t. 
    \end{equation}
    Every irrep into which the representation $\pi_t$ from~\eqref{eq:tensor-power-rep} decomposes is isomorphic to $\zeta_\mu$ for  some $\mu \in \mathcal{M}(d, t)$ and for every $\mu \in \mathcal{M}(d, t)$, there exists some subrepresentation of $\pi_t$ that is isomorphic to $\zeta_\mu$~\cite{Roy_2009}.

    We now briefly introduce the concept of a Gelfand pair following the formulation provided in~\cite{bannai2020}. Given a compact Hausdorff group $G$ and a closed subgroup $K \leq G$, we say that $(G, K)$ is a \textit{Gelfand pair} if for every unitary irrep $(V, \zeta)$ of $G$, the subspace
    \begin{equation}
        V^K \coloneqq \{|v\rangle \in V: \zeta(k)|v\rangle = |v\rangle \text{ for all }k\in K\}
    \end{equation}
    is at most one-dimensional. If $V^K$ is one-dimensional, then the irrep $(V, \zeta)$ is said to be \textit{$K$-spherical}. Given an $m$-dimensional subspace $W$ of $\h$, let us identify $U(W) \times U(W^\perp)$ with its canonical embedding into $U(\h)$: 
    \begin{equation}
        \left\{ u \oplus v: u\in U(W), v\in U(W^\perp) \right\}.
    \end{equation}
    In this case, it is known that $(U(\h), U(W) \times U(W^\perp))$ is a Gelfand pair~\cite{bannai2020}. 
    For $m\leq \lfloor d/2\rfloor$, let $\msph(d, m)$ denote the set of $\mu \in \mathcal{M}(d)$ of the form
    \begin{equation}
        (\mu_1, \dots, \mu_m, 0, \dots, 0, -\mu_m, \dots, -\mu_1) 
    \end{equation}
    and let $\msph(d, m, t) \coloneqq \msph(d, m) \cap \m(d, t)$, that is, we require $\sum_{j=1}^m \mu_j \leq t$. 
    The $U(\h)$-irrep $(V_\mu, \zeta_\mu)$ is $U(W) \times U(W^\perp)$-spherical for every $\mu \in \msph(d, m)$ and every $U(W)\times U(W^\perp)$-spherical irrep of $U(\h)$ must be isomorphic to $(V_\mu, \zeta_\mu)$ for some $\mu \in \msph(d, m)$~\cite{bannai2020}. Thus, $\msph(d, m, t)$ indexes the collection of all inequivalent $U(W) \times U(W^\perp)$-spherical irreps appearing as subrepresentations of $\pi_t$.

    Given a Gelfand pair $(G, K)$ and $K$-spherical irrep $(V, \zeta)$ of $G$, we define the \textit{normalized zonal spherical function} for $(V, \zeta)$ to be the matrix coefficient $Z_\zeta^{(G, K)}\colon G \to \mathbb{C}$ given by
    \begin{equation}\label{eq:zonal-spher-func-def}
        Z_\zeta^{(G, K)}(g) = \langle v|\zeta(g)|v\rangle,
    \end{equation}
    where $|v\rangle$ is a $K$-invariant unit vector in $V$; note that the function $Z_\zeta^{(G, K)}$ does not depend on the choice of $|v\rangle$. For the remainder of this paper, we will denote the zonal spherical function of the Gelfand pair $(U(\h), U(W)\times U(W^\perp))$ corresponding to the irrep $(V_\mu, \zeta_\mu)$ by $Z_\mu$.

    \subsection{Properties of unitary designs}

    Given an arbitrary group $K$ and weighted $K$-subsets $(\x, p), (\y, q)$, let $p\star q\colon \x\y\to [0,1]$ denote the \textit{convolution} $p\star q$ of $p$ and $q$ by 
    \begin{equation}
        (p\star q)(g) = \sum_{\stackrel{(x, y) \in \x \times \y\colon}{xy = g}} p(x)q(y).
    \end{equation} 
    Note that the convolution $p\star q$ is a pmf over $K$. 
    Since the convolution of pmf's is associative, we may define the convolution $p_1\star \dots\star p_k$ for weighted $K$-subsets $(\x_1, p_1), \dots, (\x_k, p_k)$ without ambiguity. In this case, it is straightforward to check (see Appendix~\ref{app:associativity} for more details) that 
    \begin{equation}
        (p_1 \star \dots \star p_k)(g) = \sum_{\stackrel{(x_1, \dots, x_k) \in \x_1\times \dots \times \x_k\colon}{x_1\dots x_k = g}} p_1(x)\dots p_k(x).
    \end{equation}

    We will use the following observation throughout our work. 
    \begin{lemma}\label{lemma:design-with-non-design}
        Given some closed subgroup $K \leq U(\h)$, suppose that $(\x, p)$ is a $t$-design in $K$, and let $(\y, q)$ be an arbitrary weighted $K$-subset. Then $(\x\y, p \star q)$ and $(\y\x, q \star p)$ are $t$-designs in $K$. 
    \end{lemma}
    \begin{proof}
        Note that 
        \begin{align}
            \sum_{g\in \x\y} (p\star q)(g) g^{\otimes t} \otimes \overline{g}^{\otimes t} 
            &= \sum_{g\in \x\y} \sum_{\stackrel{(x,y) \in \x\times\y:}{xy = g}} p(x)q(y) (xy)^{\otimes t} \otimes (\overline{xy})^{\otimes t} \\
            &= \sum_{x\in \x} \sum_{y\in \y} p(x)q(y) (xy)^{\otimes t} \otimes (\overline{xy})^{\otimes t} \\
            &= \sum_{y\in \y} q(y) \left(\sum_{x\in \x} p(x) x^{\otimes t} \otimes \overline{x}^{\otimes t} \right) (y^{\otimes t} \otimes \overline{y}^{\otimes t}) \\
            &= \sum_{y\in \y} q(y) \left(\int_K dk \ k^{\otimes t} \otimes \overline{k}^{\otimes t}\right)(y^{\otimes t}\otimes \overline{y}^{\otimes t}) \\
            &= \sum_{y\in \y} q(y) \int_K dk \ (ky)^{\otimes t} \otimes (\overline{ky})^{\otimes t} \\
            &= \int_K dk \ k^{\otimes t} \otimes \overline{k}^{\otimes t},
        \end{align}
        so that $(\x\y, p\star q)$ is a $t$-design in $K$. By a similar argument, $(\y\x, q \star p)$ is also a $t$-design in $K$.
    \end{proof}

    We will also use some other useful characterizations of the $t$-design property that have been employed extensively in~\cite{bannai2020, Roy_2009}:
    \begin{lemma}[{\cite{bannai2020,Roy_2009}}]
    \label{lemma:t-des-equiv-properties}
        The following hold:
        \begin{enumerate}
            \item Let $K\leq U(\h)$ be a closed subgroup. A weighted $K$-subset $(\x, p)$ is a $t$-design in $K$ iff
            \begin{equation}
                \sum_{g\in \x} p(g) \zeta_\nu(g) = \int_K dk \ \zeta_\nu(k) 
            \end{equation}
            for all $\nu \in \m(d, t)$.
            \item A weighted $U(\h)$-subset $(\x, p)$ is a $t$-design in $U(\h)$ if and only if 
            \begin{equation}
                \sum_{g\in \x} p(g) \zeta_\nu(g) = 0
            \end{equation}
            for all $\nu \in \m(d, t)\setminus\{\mathbf{0}\}$.
        \end{enumerate}
    \end{lemma}

    We will need one more technical lemma, which we will use to construct $t$-designs in $U(\h)$. The following lemma uses the central idea of~\cite{bannai2020}, though we state it in a slightly different form.
    \begin{lemma}\label{lemma:irrep-average-formula}
        Let $t\geq 0$. Given a nontrivial $m$-dimensional subspace $W \subset \h$ with $m \leq \lfloor d/2\rfloor$, let us define $K = U(W) \times U(W^\perp)$. For all $\mu \in \msph(d, m, t)$, let $Z_\mu$ be the zonal spherical function for $(V_\mu, \zeta_\mu)$ with respect to the Gelfand pair $(U(\h), K)$. Let $(\y, r)$ be a $t$-design in $K$, which exists by Proposition~\ref{prop:block-diag-designs} in \Cref{app:block-diag-designs}. Let $u_1, \dots, u_k \in U(\h)$ be arbitrary unitaries, define the set 
        \begin{equation}
            \z \coloneqq \y u_1 \y u_2 \dots \y u_{k}\y \subset U(\h),
        \end{equation}
        and endow it with the probability distribution
        \begin{equation}
            \eta \coloneqq r \star 1_{u_1} \star r \star 1_{u_2} \star \dots \star 1_{u_k} \star r,
        \end{equation}
        where $1_{u_i}$ denotes the pmf over the singleton set $\{u_i\}$. Then for all $\nu \in \m(d, t)$, we have
        \begin{align}
            \sum_{g\in \z} \eta(g) \zeta_\nu(g) 
            &= \begin{cases} 
               0, & \nu \in \m(d, t) \setminus \msph(d,m, t)\\
               \left(\prod_{i=1}^k Z_\nu(u_{i})\right) |v_\nu\rangle \langle v_\nu|, & \nu \in  \msph(d,m, t).\label{eq:projector-piecewise}
            \end{cases}
        \end{align}
    \end{lemma}

    \begin{proof}
        For all $\nu \in \m(d)$, recall that 
        \begin{equation}
            V_\nu^K \coloneqq \{|v\rangle \in V: \zeta_\nu(k)|v\rangle = |v\rangle \text{ for all }k\in K\}
        \end{equation}
        and let $\Pi_\nu^K \colon V_\nu \to V_\nu$ denote the orthogonal projector onto $V_\nu^K$. Since $(U(\h), K)$ is a Gelfand pair, we have $\dim V_\nu^K \leq 1$ for all $\nu \in \m(d)$. If $\dim V_\nu^K = 1$, let $|v_\nu\rangle$ be a $K$-invariant unit vector in $V_\nu$. With this notation, we may write 
        \begin{equation}
            \Pi_\nu^K =
            \begin{cases} 
               0 & \nu \in \m(d, t)\setminus \msph(d,m, t)\\
               |v_\nu\rangle\langle v_\nu| & \nu \in \msph(d,m, t) 
            \end{cases}
        \end{equation}
        in light of our earlier discussion on spherical irreps. For all $\nu \in \m(d, t)$, we compute:
        \begin{align}
            &\sum_{g\in \z} \eta(g) \zeta_\nu(g) \notag\\
            &= \sum_{y_1 \in \y} \dots \sum_{y_{k+1} \in \y} r(y_1) \dots r(y_{k+1}) \zeta_\nu(y_1 u_{1} y_2 u_{2} \dots y_ku_{k} y_{k+1}) \\
            &= \left(\sum_{y_1 \in \y} r(y_1) \zeta_\nu(y_1) \right) \zeta_\nu(u_1 )\left(\sum_{y_2 \in \y} r(y_2) \zeta_\nu(y_2) \right) \zeta_\nu(u_2) \dots \\
            &\ \ \ \ \ \ \ \ \ \dots \left(\sum_{y_k \in \y} r(y_k) \zeta_\nu(y_k) \right) \zeta_\nu(u_k) \left(\sum_{y_{k+1} \in \y} r(y_{k+1}) \zeta_\nu(y_{k+1}) \right) \\
            &= \left(\int_K dk \ \zeta_\nu(y_k) \right) \zeta_\nu(u_1 )\left(\int_K dk \ \zeta_\nu(y_k)  \right) \zeta_\nu(u_2) \dots \left(\int_K dk \ \zeta_\nu(y_k)  \right) \zeta_\nu(u_k) \left(\int_K dk \ \zeta_\nu(y_k) \right) \\
            &= \Pi_\nu^K \zeta_\nu(u_1 )\Pi_\nu^K \zeta_\nu(u_2) \dots \Pi_\nu^K \zeta_\nu(u_k) \Pi_\nu^K \label{eq:product-of-projectors}
        \end{align}
        The first equality above follows from the fact that $1_{u_i}(u_i) = 1$ and the properties of the $\star$-product (see the computation in Appendix~\ref{app:associativity}), and the third equality follows from the first part of Lemma~\ref{lemma:t-des-equiv-properties}. 
        The expression in \eqref{eq:product-of-projectors} is 0 if $\nu \in \m(d, t) \setminus \msph(d, m, t)$, and equal to
        \begin{align}
            |v_\nu\rangle\langle v_\nu|\zeta_\nu(u_1)|v_\nu\rangle\langle v_\nu|\zeta_\nu(u_2) |v_\nu\rangle \dots \langle v_\nu|\zeta_\nu(u_k)|v_\nu\rangle\langle v_\nu| = \left(\prod_{i=1}^k Z_\nu(u_i)\right) |v_\nu\rangle \langle v_\nu|
            \label{eq:projector-piecewise}
        \end{align}
        if $\nu \in  \msph(d, m, t)$, concluding the proof.
    \end{proof}

    \subsection{Embedding symmetric designs in unitary designs}   
    
    In the following, we assume the notation of Section~\ref{sec:defs-and-prelims}. In particular, recall that $S\leq U(\h)$ is a closed subgroup and that $\rho$ and $\sigma$ denote the natural representations of $S$ and $U(\h)_S$ on $\h$ given by $s\mapsto s$ and $u\mapsto u$, respectively. Recall that the irrep basis of $S$ is an orthonormal basis 
        \begin{equation}
            \beta \coloneqq \{|\lambda, i, j\rangle: \lambda \in [n], i\in [\dim V_\lambda], j \in [\dim W_\lambda]\}
        \end{equation}
        of $\h$ such that for all $\lambda \in [n], j \in [\dim W_\lambda]$, the subspaces
        \begin{equation}
            V_{\lambda, j} \coloneqq \mathrm{span}\left\{ |\lambda, i, j\rangle: i \in [\dim V_\lambda]\right\} 
        \end{equation}
        are irreducible under $\rho$ and for all $\lambda \in [n], i \in [\dim V_\lambda]$, the subspaces
        \begin{equation}\label{eq:multiplicity-copies}
            W_{\lambda, i} \coloneqq \mathrm{span}\left\{ |\lambda, i, j\rangle: j \in [\dim W_\lambda] \right\}
        \end{equation}
        are invariant under the natural representation $\sigma$ of $U(\h)_S$ on $\h$, and  $V_{\lambda, i} \cong V_\lambda$ and $W_{\lambda, j} \cong W_\lambda$ as $S$ and $U(\h)_S$-representations, respectively. 
    
    \begin{proposition}\label{thm:zonal-spher-fncs-zeros}
        Let $1\leq t \leq t'$. Suppose that $U(\h)_S \neq U(\h)$. Let $\lambda_0\in [n]$ be such that $\dim W_{\lambda_0}$ is the smallest of all the multiplicity space dimensions. Let us define $K = U(W_{\lambda_0, 1}) \times U(W_{\lambda_0, 1}^\perp)$ and recall that $(U(\h), K)$ is a Gelfand pair. For all $\mu \in \msph(d, \dim W_{\lambda_0})$, let $Z_\mu$ be the zonal spherical function for $(V_\mu, \zeta_\mu)$ with respect to the Gelfand pair $(U(\h), K)$. Assume that there exists a $\mu_{\mathrm{nz}} \in \msph(d, m, t+1)\setminus \{\mathbf{0}\}$ and collection of unitaries $\{u_\mu\}_{\mu \in \msph(d, m, t)\setminus \{\mathbf{0}\}}\subset U(\h)$ such that $Z_\mu(u_\mu) = 0$ while $Z_{\mu_{\mathrm{nz}}}(u_\mu)\neq 0$ for all $\mu \in \msph(d, m, t)\setminus\{\mathbf{0}\}$. Then there exists a weighted $U(\h)$-subset $(\x, p)$ such that 
        \begin{equation}
            t_{\max}^{(S)}(\x_S, p_S) \geq t' \geq t = t_{\max}(\x, p).
        \end{equation}
    \end{proposition}

    \begin{proof}
        
        Let $(\w, q)$ be a $t'$-design in $U(\h)_S$ and let $(\mathcal{Y}, r)$ be a $(t+1)$-design in $K$, which exists by Proposition~\ref{prop:block-diag-designs} in Appendix~\ref{app:block-diag-designs}.
        Let us enumerate the elements of $\msph(d, m, t)\setminus\{\mathbf{0}\}$ as $\mu_1,\dots,\mu_k$, with $k = |\msph(d, m, t)\setminus \{\mathbf{0}\}|$. Define the set
        \begin{equation}\label{eq:z-set-def}
            \mathcal{Z} \coloneqq \mathcal{Y}u_{\mu_1}\mathcal{Y}u_{\mu_2}\dots \mathcal{Y}u_{\mu_k}\mathcal{Y} \subset U(\h)
        \end{equation}
        and endow it with the distribution 
        \begin{equation}\label{eq:eta-dist-def}
            \eta\coloneqq r\star 1_{\mu_1} \star r \star 1_{\mu_2} \star \dots \star r \star 1_{\mu_k} \star r,
        \end{equation}
        where $1_{\mu_i}$ denotes the distribution over the singleton set $\{u_{\mu_i}\}$. Define $\z^{-1}$ to be the set of inverses of elements from $\z$ and define the probability distribution $\eta_{\mathrm{inv}}$ over $\z^{-1}$ by $\eta_{\mathrm{inv}}(z) = \eta(z^{-1})$. Finally, construct the set 
        \begin{equation}
            \x \coloneqq \z^{-1}\z \w
        \end{equation}
        and endow it with the distribution
        \begin{equation}
            p \coloneqq \eta_{\mathrm{inv}} \star \eta \star q.
        \end{equation}
        We involve the inverse set $\z^{-1}$ and the inverse distribution $\eta_{\mathrm{inv}}$ simply to ensure that $\x_S$ is nonempty and to facilitate calculation with $p_S$.
        
        First, we establish that $\mathcal{X}_S$ is an $S$-symmetric $t'$-design under the distribution $p_S$. For the sake of convenience, let us write $\x= \x'\w$ and $p = p' \star q$, where $\mathcal{X}' \coloneqq \z^{-1}\z$ and $p' \coloneqq  \eta_{\mathrm{inv}}\star \eta$ is a distribution over $\mathcal{X}'$. First, let $z_0\in \z$ be such that $\eta(z_0) > 0$. Thus, we have $I = z_0^{-1}z_0 \in \x'$, so that 
        \begin{equation}
            \sum_{x' \in \x'_S} p'(x') \geq p'(I) = \sum_{\stackrel{(z_1, z_2) \in (\z_{\mathrm{inv}}, \z):}{z_1z_2 = I}} \eta_{\mathrm{inv}}(z_1)\eta(z_2) \geq \eta_{\mathrm{inv}}(z_0^{-1})\eta(z_0) = \eta(z_0)^2 > 0.  \label{eq:XS'-non-empty-well-defined} 
        \end{equation}
        Thus, $p'$ endows $\x'_S$ with a non-zero probability weight, so that $p'_S$ is well-defined via~\eqref{eq:ps-def}. Suppose $g \in \mathcal{X}_S$. We may write $g = xw$ for some $x\in \mathcal{X}'$, $w\in \w$. Recalling that $\w \subset U(\h)_S$, the fact that both $xw$ and $w$ commute with every element of $S$ implies that $x$ commutes with all of $S$. Hence, $\mathcal{X}_S = \mathcal{X}'_S \w$. We also have
        \begin{align}
            p_S(g) &\coloneqq p(g)\left(\sum_{g' \in \mathcal{X}_S} p(g') \right)^{-1} \\
            &= \left(\sum_{\stackrel{(x,w)\in \mathcal{X}'\times \w:}{xw = g}} p'(x)q(w) \right) \left(\sum_{g' \in \mathcal{X}_S}  \sum_{\stackrel{(x, w)\in\mathcal{X}'\times \w:}{xw = g'}} p'(x)q(w)\right)^{-1} \\
            &= \left(\sum_{\stackrel{(x,w)\in \mathcal{X}'\times \w:}{xw = g}} p'(x)q(w) \right) \left(\sum_{g' \in \mathcal{X}_S}  \sum_{\stackrel{(x, w)\in\mathcal{X}'_S\times \w:}{xw = g'}} p'(x)q(w)\right)^{-1} \label{eq:X-XS-1}\\
            &= \left(\sum_{\stackrel{(x,w)\in \mathcal{X}'\times \w:}{xw = g}} p'(x)q(w) \right) \left(\sum_{x\in \mathcal{X}_S'} \sum_{w\in \w} p'(x)q(w)\right)^{-1} \label{eq:XS-XS'W} \\
            &= \left(\sum_{\stackrel{(x,w)\in \mathcal{X}'\times \w:}{xw = g}} p'(x)q(w) \right) \left(\sum_{x\in \mathcal{X}_S'} p'(x)\right)^{-1} \\
            &= \sum_{\stackrel{(x,w)\in \mathcal{X}'\times \w:}{xw = g}} p'(x) \left(\sum_{x'\in \mathcal{X}_S'} p'(x')\right)^{-1} q(w) \\
            &= \sum_{\stackrel{(x,w)\in \mathcal{X}'\times \w:}{xw = g}} p'_S(x) q(w) \\
            &=
            \sum_{\stackrel{(x, w)\in \mathcal{X}'_S \times \w:}{xw = g}} p'_S(x) q(w) \label{eq:X-XS-2}  \\
            &= (p'_S \star q)(g). 
        \end{align}
        The equalities in \eqref{eq:X-XS-1} and \eqref{eq:X-XS-2} follow again from the fact that if $xw$ and $w$ commute with all of $S$, then so must $x$. Equality \eqref{eq:XS-XS'W} follows from $\x_S = \x_S' \w$.
        Since $p_S = p'_S \star q$ and $\mathcal{X}_S = \mathcal{X}'_S \w$, we may conclude using Lemma~\ref{lemma:design-with-non-design} and the fact that $(\w, q)$ is a $t'$-design in $U(\h)_S$ that $(\x_S, p_S)$ is also a $t'$-design in $U(\h)_S$.

        We now argue that $(\x, p)$ is a $t$-design in $U(\h)$. By assumption, $u_{\mu}$ is a zero of $Z_{\mu}$ for all $\mu \in \msph(d, m, t)\setminus\{\mathbf{0}\}$, so that $\prod_{i=1}^k Z_\nu(u_{\mu_i}) = 0$ for all $\nu \in \msph(d, m, t)\setminus\{\mathbf{0}\}$. It then follows from Lemma~\ref{lemma:irrep-average-formula} that 
        \begin{equation}
            \sum_{g\in \z} \eta(g) \zeta_\nu(g) = 0
        \end{equation}
        for all $\nu \in \m(d, t)\setminus\{\mathbf{0}\}$. Part 2 of Lemma~\ref{lemma:t-des-equiv-properties} then allows us to conclude that $(\z, \eta)$ is a $t$-design in $U(\h)$. We now deduce from Lemma~\ref{lemma:design-with-non-design} that $(\x, p) = (\z^{-1}\z\w, \eta_{\mathrm{inv}}\star \eta \star q)$ is a $t$-design in $U(\h)$. 
        
        Finally, we argue that $(\x, p)$ is not a $(t+1)$-design in $U(\h)$. By assumption, we have $Z_{\mu_{\mathrm{nz}}}(u_{\mu_i}) \neq 0$ for \textit{all} $1\leq i \leq k$. Therefore, Lemma~\ref{lemma:irrep-average-formula} yields \begin{equation}
            \sum_{g\in \z} \eta(g) \zeta_{\mu_{\mathrm{nz}}}(g) = c |v_{\mu_{\mathrm{nz}}}\rangle\langle v_{\mu_{\mathrm{nz}}}|,
        \end{equation}
        for some $c\neq 0$ and unit vector $|v_{\mu_{\mathrm{nz}}}\rangle \in V_{\mu_{\mathrm{nz}}}^K$, so that
        \begin{equation}
            \sum_{g\in \z^{-1}} \eta_{\mathrm{inv}}(g) \zeta_{\mu_{\mathrm{nz}}}(g) = \sum_{g\in \z} \eta(g) \zeta_{\mu_{\mathrm{nz}}}(g^{-1}) = \Bigg(\sum_{g\in \z} \eta(g) \zeta_{\mu_{\mathrm{nz}}}(g)\Bigg)^\dagger = \overline{c} |v_{\mu_{\mathrm{nz}}}\rangle\langle v_{\mu_{\mathrm{nz}}}|,
        \end{equation}
        where $\overline{c}$ denotes the complex conjugate of $c$.
        
        Let $\Pi_{\mu_{\mathrm{nz}}}^{U(\h)_S}\colon V_{\mu_{\mathrm{nz}}}\to V_{\mu_{\mathrm{nz}}}$ denote the orthogonal projector onto the subspace of $V_{\mu_{\mathrm{nz}}}$ that is invariant under $U(\h)_S$. Then
        \begin{align}
            \sum_{g\in \x} p(g) \zeta_{\mu_{\mathrm{nz}}}(g) &= \sum_{z_1 \in \z^{-1}} \sum_{z_2 \in \z} \sum_{w\in \w} \eta_{\mathrm{inv}}(z_1)\eta(z_2)q(w) \zeta_{\mu_{\mathrm{nz}}}(z_1z_2w) \\
            &= \left( \sum_{z_1 \in \z^{-1}} \eta_{\mathrm{inv}}(z_1) \zeta_{\mu_{\mathrm{nz}}}(z_1)\right)\left( \sum_{z_2 \in \z} \eta(z_2) \zeta_{\mu_{\mathrm{nz}}}(z_2)\right) \left( \sum_{w\in \w} q(w) \zeta_{\mu_{\mathrm{nz}}}(w) \right) \\
            &= \left( \sum_{z_1 \in \z^{-1}} \eta_{\mathrm{inv}}(z_1) \zeta_{\mu_{\mathrm{nz}}}(z_1)\right)\left( \sum_{z_2 \in \z} \eta(z_2) \zeta_{\mu_{\mathrm{nz}}}(z_2)\right) \left( \int_{U(\h)_S} du \ \zeta_{\mu_{\mathrm{nz}}}(u)\right) \\
            &= \left(\overline{c}|v_{\mu_{\mathrm{nz}}}\rangle\langle v_{\mu_{\mathrm{nz}}}|\right)\left(c|v_{\mu_{\mathrm{nz}}}\rangle\langle v_{\mu_{\mathrm{nz}}}|\right) \Pi_{\mu_{\mathrm{nz}}}^{U(\h)_S} \\
            &= |c|^2 |v_{\mu_{\mathrm{nz}}}\rangle\langle v_{\mu_{\mathrm{nz}}}| \\
            &\neq 0,
        \end{align}
        where the third equality is due to Part 1 of Lemma~\ref{lemma:t-des-equiv-properties} and the last equality is due to the fact that $U(\h)_S \subset K$, so that $\Pi_{\mu_{\mathrm{nz}}}^{U(\h)_S}|v_{\mu_{\mathrm{nz}}}\rangle = |v_{\mu_{\mathrm{nz}}}\rangle$. Therefore, $(\x, p)$ is not a $(t+1)$-design in $U(\h)$ by Lemma~\ref{lemma:t-des-equiv-properties}.       
    \end{proof}

    The construction in Proposition~\ref{thm:zonal-spher-fncs-zeros} for designs whose strength is enhanced by a symmetry is based on understanding the zeros of zonal spherical functions. Using Lemma~\ref{lemma:common-zeros-1} below, we can study the zeros of zonal spherical functions in a special case.  
    \begin{lemma}\label{lemma:common-zeros-1}
        Let $t\geq 1$. Let $W$ be an arbitrary \textit{one-dimensional} subspace of $\h$. Let $K = U(W) \times U(W^\perp) \leq U(\h)$. For $t\geq 0$, let $Z_t$ denote the zonal spherical function $Z_\mu$ of the Gelfand pair $(U(\h), K)$ corresponding to $\mu = (t, 0, \dots, 0, -t) \in \msph(d, 1)$. Then the following hold.
        \begin{enumerate}
            \item There exist at most finitely many values of $d\coloneqq \dim \h \geq 1$ for which there exist a unitary $u\in U(\h)$ and an integer $1\leq s \leq t$ such that $Z_s(u) = 0 = Z_{t+1}(u)$.
            \item In particular, when $1 \leq t \leq 3$ and $d > 2$, there do not exist a unitary $u\in U(\h)$ and an integer $1 \leq s \leq t$ for which $Z_s(u) = 0 = Z_{t+1}(u)$. 
        \end{enumerate}  
    \end{lemma}

    We provide a more detailed exposition on the structure of zonal spherical functions in Section~\ref{sec:more-zonal-spher-func} and defer a proof of Lemma~\ref{lemma:common-zeros-1} until Section~\ref{sec:common-zeros}.
    We are now ready to prove the following main result of this work:

    \begin{theorem}[Theorem~\ref{thm:main-thm-1} restated]\label{thm:main-thm-1-restated}
    Let $t' \geq t> 0$ be arbitrary. There exists a $d_0>0$ such that for any Hilbert space $\h$ of dimension $\dim \h \geq d_0$ and closed subgroups $S \leq U(\h)$ for which the natural representation of $S$ on $\h$ has a multiplicity-free irreducible subrepresentation, there exists a weighted $U(\h)$-subset $(\x, p)$ satisfying
    \begin{align}
        t_{\max}^{(S)}(\x, p) \geq t' > t = t_{\max}(\x, p). 
    \end{align}
    In particular, if $1\leq t \leq 3$, then we may take $d_0 = 3$ in the preceding statement.
    \end{theorem}

    \begin{proof}
        Let $\lambda_0 \in [n]$ be such that $W_{\lambda_0}$ has dimension one; such a $\lambda_0$ must exist by hypothesis. Now let 
        \begin{equation}
            K = U(W_{\lambda_0, 1}) \times U(W_{\lambda_0, 1}^\perp) \coloneqq \left\{ u \oplus v: u \in U(W_{\lambda_0, 1}), v\in U(W_{\lambda_0, 1}^\perp) \right\},
        \end{equation}
        where $W_{\lambda_0, 1}$ is defined by~\eqref{eq:multiplicity-copies}. For $\mu \in \msph(d, 1)$, let $Z_\mu$ denote the zonal spherical function for the irrep $(V_\mu, \zeta_\mu)$ of the Gelfand pair $(U(\h), K)$ and let $|v_\mu\rangle$ denote a unit vector in $V_\mu^K$. Note that for all $\mu \in \msph(d, 1)\setminus\{\mathbf{0}\}$, we have
        \begin{equation}
            \int_{U(\h)} du \ Z_\mu(u) = \langle v_\mu|\left(\int_{U(\h)} du \ \zeta_\mu(u)\right)|v_\mu\rangle = 0,
        \end{equation}
        because $\int_{U(\h)} du \ \zeta_\mu(u)$ is the orthogonal projector onto the $U(\h)$-invariant subspace off $V_\mu$, which contains only the zero vector. As pointed out in~\cite{bannai2020}, it follows from the intermediate value theorem that for all $\mu \in \msph(d, 1)\setminus\{\mathbf{0}\}$, the zonal spherical function $Z_\mu$ has a zero in $U(\h)$. It then follows from the first part of Lemma~\ref{lemma:common-zeros-1} that if $d$ is sufficiently large, any zero of $Z_\mu$ for $\mu \in \msph(d, 1, t)\setminus \{\mathbf{0}\}$ cannot be a zero of $Z_{\mu_{\mathrm{nz}}}$, where $\mu_{\mathrm{nz}} \coloneqq (t+1, 0, \dots, 0, -t-1) \in \msph(d, 1, t+1)\setminus\msph(d, 1, t)$. It also follows from the second part of Lemma~\ref{lemma:common-zeros-1} that the same conclusion holds if $d>2$ and $1\leq t \leq 3$. The statement of the theorem then follows from Proposition~\ref{thm:zonal-spher-fncs-zeros}.
    \end{proof}

    Theorem~\ref{thm:main-thm-1-restated} applies to certain important physical symmetry groups. We list a few examples. To do so, let us introduce the operator $P_d(\pi)$ for permutation $\pi \in \mathfrak{S}_N$, which commutes the $N$ tensor factors in $(\mathbb{C}^d)^{\otimes N}$: 
    \begin{equation}
        P_d(\pi)|v_1\rangle \otimes \dots \otimes |V_N\rangle = |v_{\pi^{-1}(1)}\rangle \otimes \dots \otimes |v_{\pi^{-1}(N)}\rangle. 
    \end{equation}
    \begin{corollary}
        Let $N\geq 1$, $d\geq 2$. Let $\h_1 = (\mathbb{C}^d)^{\otimes N} = \h_3$, and $\h_2 = (\mathbb{C}^2)^{\otimes N}$. Define subgroups $S_i \leq U(\h_i), i = 1,2,3$ by 
        \begin{align}
            S_1 &\coloneqq \left\{u^{\otimes N}: u \in \mathrm{SU}(d)\right\} \cong \mathrm{SU}(d)  \\
            S_2 &\coloneqq \left\{(e^{i\theta Z})^{\otimes N}: \theta \in \mathbb{R}\right\} \cong \mathrm{U}(1)\\
            S_3 &\coloneqq \left\{P_d(\pi): \pi \in \mathfrak{S}_N\right\} \cong \mathfrak{S}_N, \\
        \end{align}
        where $Z$ denotes the Pauli $Z$ operator on $\mathbb{C}^2$. 
        Let $t' \geq t \geq 1$ be arbitrary integers. For $i = 1,2$ and all $N$  sufficiently large, there exists a weighted $U(\h_i)$-subset $(\x^{(i)}, p^{(i)})$ such that 
        \begin{equation}\label{eq:t-max-ineqs-2}
            t_{\max}^{(S_i)}(\x^{(i)}_{S_i}, p^{(i)}_{S_i}) \geq t' \geq t = t_{\max}(\x^{(i)}, p^{(i)}).  
        \end{equation}
        In particular, for $i = 1,2$, $N\geq 2$, and $1\leq t \leq 3$, there exists a weighted $U(\h_i)$-subset $(\x, p)$ for which~\eqref{eq:t-max-ineqs-2} holds. For $i= 3$ and all $N$ that are sufficiently large multiples of $d$, there exists a weighted $U(\h_i)$-subset $(\x, p)$ such that~\eqref{eq:t-max-ineqs-2} holds. In particular, this inequality is achieved if $1\leq t\leq 3$ and $N$ is any positive multiple of $d$. 
    \end{corollary}

    \begin{proof}
        By the Schur-Weyl duality, the space $(\mathbb{C}^d)^{\otimes N}$ decomposes under the action of $S_1$ and $S_3$ as
        \begin{equation}\label{eq:schur-weyl-decomp}
            (\mathbb{C}^d)^{\otimes N} \cong \bigoplus_{\lambda \vdash_d N} V_\lambda \otimes W_\lambda,
        \end{equation}
        where the direct sum runs over all partitions of $N$ into at most $d$ parts, the $V_\lambda$ are the irreps of $S_1$, and the $W_\lambda$ are the irreps of $S_3$. By the hook length formula, the $S_1$-irrep $V_{(N)}$ occurs with multiplicity $\dim W_{(N)} = 1$. Furthermore, if $N$ is a multiple of $d$, then the $S_3$-irrep $W_{(N/d, \dots, N/d)}$ occurs with multiplicity $\dim V_{(N/d, \dots, N/d)} = 1$ by the Weyl dimension formula. Finally, the space $(\mathbb{C}^2)^{\otimes N}$ decomposes into irreps under the action of $S_2$ as
        \begin{equation}
            (\mathbb{C}^2)^{\otimes N} \cong \bigoplus_{w = 0}^N V_w \otimes W_w,
        \end{equation}
        where $V_w$ is a one-dimensional irrep of $S_2$ spanned by a computational basis vector with Hamming weight $w$ and occurs with multiplicity $\dim W_w = {N \choose w}$. Thus, the irrep $V_0$ has multiplicity one. 
        The statement of the corollary now follows from Theorem~\ref{thm:main-thm-1-restated}.
    \end{proof}

    \subsection{More on zonal spherical functions}\label{sec:more-zonal-spher-func}

    We now study the zeros of the zonal spherical functions. The following exposition closely follows that of Kurihara and Okuda~\cite{kurihara2013} and Roy~\cite{roy2008}.

    \subsubsection{Principal angles}

    For $1 \leq m \leq \lfloor d/2\rfloor$, let $G_{m,d}$ denote the Grassmannian of $m$-dimensional subspaces in $\h$. Let $W$ be a fixed $m$-dimensional subspace of $\h$ and note that there is a bijective correspondence between the coset space $U(\h)/(U(W) \times U(W^\perp))$ and the Grassmannian $G_{m, d}$ via the map $u(U(W) \times U(W^\perp)) \mapsto uW$. 

    Given two subspaces $a, b \in G_{m, d}$, we define a sequence of parameters known as the \textit{principal angles} $\theta_1, \dots, \theta_m$, which capture a notion of the distance between $a$ and $b$. The first angle $\theta_1$ is defined as
    \begin{equation}
        \theta_1(a, b) \coloneqq \min\left\{\arccos |\langle x_1|y_1\rangle|: |x_1\rangle \in a, |y_1\rangle \in b, \langle x_1|x_1\rangle = 1 = \langle y_1|y_1\rangle\right\};
    \end{equation}
    let $|x_1\rangle \in a, |y_1\rangle \in b$ be unit vectors that achieve the minimal value. 
    The second angle $\theta_2(a, b)$ is defined recursively in terms of $|x_1\rangle, |y_1\rangle$ as
    \begin{equation}
        \theta_2(a,b) = \min\left\{\arccos |\langle x_2|y_2\rangle|: |x_2\rangle \in a\cap \Span\{|x_1\rangle\}^\perp, |y_2\rangle \in b\cap \Span\{|y_1\rangle\}^\perp, \langle x_2|x_2\rangle = 1 = \langle y_2|y_2\rangle\right\}.
    \end{equation}
    By continuing the process of finding the minimal angle between vectors in increasingly restricted subspaces of $a$ and $b$, we arrive at the sequence of \textit{principal angles} $0 \leq \theta_1 \leq \theta_2 \leq \dots \leq \theta_m \leq \pi/2$. 

    For any subspace $a\in G_{m, d}$, let $P_a$ denote the orthogonal projector onto $a$. For the sake of convenience, we define 
    \begin{equation}
        y_i^2(a, b) \coloneqq \cos^2 \theta_i(a, b)
    \end{equation}
    for all $a, b \in G_{m, d}$ and $1\leq i \leq m$. Given any two subspaces $a,b \in G_{m, d}$, the values $y_1(a, b), \dots, y_m(a,b)$ are all eigenvalues of $P_a P_b$ and all other eigenvalues are zero.  

    \subsubsection{Harmonic analysis on $G_{m, d}$}

    The unitary group $U(\h)$ naturally acts on $G_{m, d}$ by unitarily rotating subspaces. In this way, we may define the action of the unitary group $U(\h)$ on $G_{m, d}\times G_{m, d}$ by $u\cdot(a, b) = (ua, ub)$. The principal angles between two subspaces fully characterize the orbits of $U(\h)$ in $G_{m, d} \times G_{m, d}$. More precisely, subspaces $a, b \in G_{m, d}$ have the same principal angles as subspaces $p, q \in G_{m, d}$ \textit{if and only if} there exists a unitary $u \in U(\h)$ such that $p = ua$ and $q = ub$~\cite{roy2008, kurihara2013}

    Since $U(\h)$ acts on $G_{m, d}$ by unitarily rotating subspaces, the Grassmannian inherits a unitarily invariant measure; integration with respect to this measure is given in terms of the Haar measure over $U(\h)$ and the fixed subspace $W$ by 
    \begin{equation}
        \int_{G_{m,d}} dp \ f(p) = \int_{U(\mathcal{H})} du \ f(uW)
    \end{equation}
    for any integrable function $f$. Let $C^0(G_{m, d})$ denote the space of all continuous complex-valued functions over $G_{m, d}$. For any two $f, g \in C^0(G_{m, d})$, we may define the inner product by
    \begin{equation}\label{eq:grassmann-inner-product}
        \langle f, g\rangle \coloneqq \int_{G_{m, d}} dp \ f(p)^\star g(p).
    \end{equation}
    The space $C^0(G_{m, d})$ also carries a representation $\alpha$ of $U(\h)$ given by
    \begin{equation}\label{eq:alpha-rep}
        (\alpha(u)f)(p) = f(u^\dagger p).
    \end{equation}

    Now let us write 
    \begin{equation}
        K \coloneqq U(W)\times U(W^\perp)
    \end{equation}
    and fix $\mu \in \msph(d, m)$. Let us define the map 
    \begin{equation}
        \Phi_\mu: V_\mu \otimes V_\mu^K \to C^0(G_{m, d})
    \end{equation}
    by 
    \begin{equation}
        \Phi_\mu(|v\rangle\otimes |w\rangle)(uW) = \langle v| \zeta_\mu(u) |w\rangle
    \end{equation}
    and observe that this definition is independent of the choice of $u$ defining the subspace $uW$; indeed, if $uW = vW$, then $u^{-1}v \in K$, so that 
    \begin{equation}
        \langle v|\zeta_\mu(u)|w\rangle = \langle v|\zeta_\mu(u)\zeta_\mu(u^{-1}v)|w\rangle =  \langle v|\zeta_\mu(uu^{-1})\zeta_\mu(v)|w\rangle = \langle v|\zeta_\mu(v)|w\rangle
    \end{equation}
    by the $K$-invariance of $V_\mu$. It may be verified that 
    \begin{equation}\label{eq:H-mu-definition}
        H_\mu \coloneqq \varphi_\mu(V_\mu \otimes V_\mu^K)
    \end{equation}
    is a finite-dimensional subspace of $C^0(G_{m, d})$.

    The Peter-Weyl theorem for compact symmetric spaces then yields the following:

    \begin{lemma}[\cite{kurihara2013}, Fact 3.4]
        The following hold:
        \begin{enumerate}
            \item For all $\mu \in \msph(d, m)$, the space $H_\mu$ carries the unique irreducible $U(\h)$-subrepresentation of $C^0(G_{m, d})$ that is  isomorphic to $V_\mu$. 
            \item If $\mu \neq \mu'$, then $H_\mu$ and $H_{\mu'}$ are orthogonal subspaces with respect to~\eqref{eq:grassmann-inner-product}.
            \item The subspace $\bigoplus_{\mu \in \msph(d, m)} H_\mu$ is dense in $C^0(G_{m, d})$. 
        \end{enumerate}
    \end{lemma}

    The \textit{zonal orthogonal polynomial $Z_{\mu, \alpha}$ of $H_\mu$ at the point $a\in G_{m, d}$} is defined to be the unique function in $H_\mu$ such that
    \begin{equation}
        \langle Z_{\mu, a}, p\rangle = p(a)
    \end{equation}
    for any $p\in H_\mu$; the existence of such a function is guaranteed by the Riesz representation theorem. Note that for any $p \in H_\mu$, $u\in U(\h)$, $a\in G_{m, d}$, we have
    \begin{align}
        \langle \alpha(u^\dagger) Z_{\mu, ua}, p\rangle = \langle \alpha(u^\dagger) Z_{\mu, ua}, \alpha(u^\dagger) \alpha(u)p\rangle = \langle Z_{\mu, ua}, \alpha(u)p\rangle = p(a) = \langle Z_{\mu, a}, p\rangle,
    \end{align}
    where we used the invariance under unitary translations of the integral defining the inner product in the second equality. 
    Since $p$ was arbitrary, it follows that $\alpha(u^\dagger) Z_{\mu, ua} = Z_{\mu, a}$, or that 
    \begin{equation}
        Z_{\mu, ua}(ub) = Z_{\mu, a}(b)
    \end{equation}for any $b\in G_{m, d}$. Hence, $Z_{\mu, a}(b)$ is constant along the orbit of $(a, b)$ in $G_{m, d} \times G_{m, d}$ under $U(\h)$. Therefore, $Z_{\mu, a}(b)$ may be written solely in terms of the principal angles between $a$ and $b$. 
    
    A general formula for the zonal orthogonal polynomials was, to the best of our knowledge, first identified in~\cite{James1974}, wherein $Z_{\mu, a}(b)$ is expressed as a real symmetric polynomial in $y_1(a, b), \dots, y_m(a, b)$. The authors of~\cite{kurihara2013} explicitly evaluated this formula for particular families of $\mu\in \msph(d, m)$, and we quote their calculation below for $\mu = (t, 0, \dots, 0, -t), t\geq 0$. To state the form of this polynomial, we recall that the \textit{Schur polynomial $X_\nu$} for any partition $\nu = (\nu_1, \dots, \nu_m)$, $\nu_1 \geq \dots \geq \nu_m \geq 0$ in variables $x_1, \dots, x_m$ is formally defined by 
    \begin{equation}
        X_\nu(x_1, \dots, x_m) \coloneqq \frac{\det(x_i^{\nu_j + m- j})_{i,j = 1}^m}{\det(x_i^{m-j})_{i,j=1}^m}
    \end{equation}
    and the \textit{normalized Schur polynomial $X_\nu$} for partition $\nu$ is defined as 
    \begin{equation}
        X_\nu^\star = \frac{1}{X_\nu(1, \dots, 1)} X_\nu.
    \end{equation}
    For the sake of notational convenience, let $(j)$ denote the $m$-tuple $(j, 0, \dots, 0)$ for any $j\geq 0$.

    \begin{lemma}\label{lemma:polynomial-zonal-orthogonal-fnc}
        Let $a, b\in G_{m, d}$. For any $\mu \in \msph(d, m)$, there exists a real symmetric polynomial $P_\mu$ in $m$ variables such that 
        \begin{equation}\label{eq:sym-poly-eq}
            Z_{\mu, a}(b) = P_\mu(y_1(a,b), \dots, y_m(a,b)).
        \end{equation}        
        Let $t\geq 0$ be arbitrary and let $\mu = (t, 0, \dots, 0, -t) \in \msph(d, 1)$. Then~\eqref{eq:sym-poly-eq} holds when $P_\mu$ is defined as
        \begin{equation}\label{eq:P-mu-quote}
            P_{\mu} \coloneqq \frac{(d+ 2t - 1){d + t - 2\choose t}^2}{(d - 1){d - m + t - 1 \choose t}} \sum_{j=0}^t (-1)^{t - j} {d + t + j - 2 \choose j}{m + t - 1 \choose t - j} X_{(j)}^\star.
        \end{equation}
    \end{lemma} 
    The formula~\eqref{eq:P-mu-quote} may be found in Remark 4.10 of~\cite{kurihara2013}.

    Finally, we concretely establish the connection between zonal orthogonal polynomials developed for $G_{m, d}$ and the zonal spherical functions of the Gelfand pair $(U(\h), K)$. Let $\mu \in \msph(d, m)$ be arbitrary, let $d_\mu \coloneqq \dim V_\mu$, and let $|v_\mu\rangle$ be a unit vector in $V_\mu^K$. Recall that the zonal spherical function $Z_\mu \colon U(\h) \to \mathbb{C}$ for the Gelfand pair $(U(\h), K)$ is defined by $Z_{\mu}(u) = \langle v_\mu|\zeta_\mu(u)|v_\mu\rangle$.
    
    \begin{lemma}\label{lemma:zonal-fnc-equivalence}
        We have 
        \begin{equation}
            Z_\mu(u) = d_\mu^{-1} Z_{\mu, W}(uW)
        \end{equation}
        for all $u\in U(\h)$.
    \end{lemma}

    \begin{proof}
        Let $\{|e_i\rangle\}_{i=1}^{d_\mu}$ be an orthonormal basis for $|e_1\rangle \coloneqq |v_\mu\rangle \in V_\mu^K$. Let $e_i\colon G_{m, d} \to \mathbb{C}$ be defined by 
        \begin{equation}
            e_i(uW) = d_\mu^{1/2} \langle e_i|\zeta_\mu(u)|v_\mu\rangle;
        \end{equation} 
        the definition of $e_i$ is independent of the choice of $u$ because $|v_\mu\rangle$ is invariant under $K$.
        By the construction of $H_\mu$ in~\eqref{eq:H-mu-definition}, the set $\{e_i\}_{i=1}^{d_\mu}$ spans $H_\mu$. Furthermore, the Schur orthogonality imply that $\{e_i\}_{i=1}^{d_\mu}$  forms an \textit{orthonormal} basis for $H_\mu$. The zonal spherical function $Z_\mu$ descends to a map $\tilde{Z}_\mu \in H_\mu$ defined by $\tilde{Z}_\mu(uW) = Z_\mu(u)$; the $K$-invariance of $|v_\mu\rangle$ again ensures that $\tilde{Z}_\mu$ is well-defined. Since $|e_1\rangle \coloneqq |v_\mu\rangle$, we have 
        \begin{equation}
            \tilde{Z}_\mu = d_\mu^{-1/2}e_1
        \end{equation}
        Letting $p\in H_\mu$ be arbitrary, the fact that $\{e_i\}_{i=1}^{d_\mu}$ is an orthonormal basis for $H_\mu$ lets us write
        \begin{equation}
            p = \sum_{i=1}^{d_\mu} \langle e_i, p\rangle e_i,
        \end{equation}
        so that
        \begin{equation}
            \langle \tilde{Z}_\mu, p\rangle = d_\mu^{-1/2} \langle e_1, p\rangle = d_\mu^{-1/2} \sum_{i=1}^{d_\mu} \delta_{i1} \langle e_i, p\rangle =  d_\mu^{-1} 
            \sum_{i=1}^{d_\mu} e_i(W) \langle e_i, p\rangle = d_\mu^{-1} p(W). 
        \end{equation}
        We then conclude that $d_\mu \tilde{Z}_\mu = Z_{\mu, W}$ by the definition of $Z_{\mu, W}$. 
    \end{proof}

    \subsubsection{Common zeros of zonal spherical functions}\label{sec:common-zeros}

    Let us briefly review the \textit{resultant} of two polynomials, following the exposition of~\cite{cox2005, cox2025}. Let $K$ be a field and suppose that $p, q$ are polynomials of positive degree in the polynomial ring $K[x]$. We may write 
    \begin{align}
        p &= a_lx^l + \dots + a_0, \ l>0 \\
        q &= b_mx^m + \dots + b_0, \ m>0
    \end{align}
    for some coefficients $a_0, \dots, a_l, b_0, \dots, b_m \in K$. The \textit{resultant of $p$ and $q$, denoted $R_{p,q}$,} is defined to be the determinant of the following $(l+m) \times (l+m)$ matrix, where blank spaces denote zeros:
    \begin{equation}
        R_{p,q} = \det 
        \begin{pmatrix}
            a_l & &  & & b_m &  &  &  \\
            a_{l-1} & a_l &  &  & b_{m-1} & b_{m} &  &  \\
            a_{l-2} & a_{l-1} & \ddots &  & b_{m-2} & b_{m-1} & \ddots &   \\
            \vdots & a_{l-2} & \ddots & a_l & \vdots & b_{m-2} & \ddots & b_m \\
             & \vdots & \ddots & a_{l-1} & & \vdots & \ddots & b_{m-1} \\
             a_0 & & & \vdots & b_0 & & & \vdots \\
             & a_0 & & & &  b_0 & & \\
             & & \ddots & & & & \ddots & & \\
             & & & a_0 & & & & b_0 
        \end{pmatrix} 
    \end{equation}
    The matrix defining the determinant has exactly $m$ columns consisting of the coefficients $a_0, \dots, l$ and exactly $l$ columns consisting of the coefficients $b_0, \dots, b_m$. As a simple example, the resultant $R_{p,q}$ of the polynomials $p = 3x^3 + 2x + 1$ and $q = x^2 + 2x + 1$ in $\mathbb{R}[x]$ is 
    \begin{equation}
        R_{p,q} = \det\begin{pmatrix}
            3 & 0 & 1 & 0 & 0 \\
            0 & 3 & 2 & 1 & 0 \\
            2 & 0 & 1 & 2 & 1 \\
            1 & 2 & 0 & 1 & 2 \\
            0 & 1 & 0 & 0 & 1 
        \end{pmatrix} = 16.
    \end{equation}
    Resultants make themselves useful for our purposes via the following lemma (see Ch.~3, $\S$6 and Prop.~3 in~\cite{cox2025}). 
    \begin{lemma}\label{lemma:common-factor-property}
        The resultant of polynomials $p, q \in K[x]$ is zero if and only if $p$ and $q$ share a common factor in $K[x]$.
    \end{lemma}

    Using Lemma~\ref{lemma:common-factor-property}, we can study the zeros of zonal spherical functions in a special case.  Let $\h$ be a $d$-dimensional Hilbert space with subspace $W$ of dimension 1. Let $K = U(W) \times U(W^\perp) \leq U(\h)$. For $t\geq 0$, let $Z_t$ denote the zonal spherical function $Z_\mu$ of the Gelfand pair $(U(\h), K)$ for $\mu = (t, 0, \dots, 0, -t) \in \msph(d, 1)$. 
    \begin{lemma}[Lemma~\ref{lemma:common-zeros-1} restated]
        Let $t\geq 1$. Then the following hold.
        \begin{enumerate}
            \item There exist at most finitely many values of $d\coloneqq \dim \h \geq 1$ for which there exist a unitary $u\in U(\h)$ and an integer $1\leq s \leq t$ such that $Z_s(u) = 0 = Z_{t+1}(u)$.
            \item In particular, when $1 \leq t \leq 3$ and $d > 2$, there do not exist a unitary $u\in U(\h)$ and an integer $1 \leq s \leq t$ for which $Z_s(u) = 0 = Z_{t+1}(u)$. 
        \end{enumerate}  
    \end{lemma}

    \begin{proof}
        By Lemmas~\ref{lemma:polynomial-zonal-orthogonal-fnc} and~\ref{lemma:zonal-fnc-equivalence}, we may write the zonal spherical function $Z_q$ for integer $q\geq 1$ as
        \begin{equation}\label{eq:m-1-zonal-spher-func}
            Z_q(u) = C_q P_{q}(|\langle e|u|e\rangle|^2),
        \end{equation}
        where $C_q$ is an irrelevant nonzero constant and $P_q$ is a real polynomial given by
        \begin{align}\label{eq:P-t-poly-def}
            P_q &\coloneqq \sum_{j=0}^q a^{(q)}_j x^j \\
            a^{(q)}_j &\coloneqq (-1)^{q-j} { d + q + j - 2 \choose j} {q \choose q - j}.       
        \end{align}
        Note that for integer $q\geq 1$, we may also write
        \begin{equation}
            a^{(q)}_j = \begin{cases}
                (-1)^{q -j}, & j = 0\\
                \frac{(-1)^{q-j}}{j!}{q\choose q-j}(d + q -1)(d + q) \dots (d+q + j - 2), & 1 \leq j \leq s.
            \end{cases}
        \end{equation} 
        Let us denote the resultant of polynomials $P_{p}, P_q$ by $R_{p,q}$ for integers $p, q\geq 1$.

        Now fix $1 \leq s \leq t$. 
        Due to the fact that the polynomials $P_{t+1}$ and $P_s$ have coefficients that have a polynomial dependence on $d$, the resultant $R_{t+1,s}$ also has a polynomial dependence on $d$. Hence, we may view $d$ as an indeterminate and write $a_j^{(s)}, a_j^{(t+1)}, R_{t+1, s} \in \mathbb{R}[d]$. Let us establish that $R_{t+1, s}$ is not the zero polynomial in $\mathbb{R}[d]$. Note that the evaluations of $a_j^{(s)}$ and $a_j^{(t+1)}$ at $1-s$ are given by 
        \begin{align}
            a_j^{(s)}(1-s) = \begin{cases}
                (-1)^{s-j}, &j = 0 \\
                0, &1\leq j \leq s
            \end{cases}
        \end{align}
        and
        \begin{align}
            a_{t+1}^{(t+1)}(1-s) = \frac{(t-s+1)(t-s + 2)\dots (2t -s + 1)}{(t+1)!},
        \end{align}
        so that the evaluation of $R_{t, s}$ at $1-s$ is 
        \begin{align}
            &R_{t+1, s}(1-s) \\
            &= 
            \det
            \begin{pmatrix}
            a_{t+1}^{(t+1)}(1-s) & &  & & &  &  &  \\
            a_{t}^{(t+1)}(1-s) & a_{t+1}^{(t+1)}(1-s) &  &  &  &  &  &  \\
            a_{t-1}^{(t+1)}(1-s) & a_{t}^{(t+1)}(1-s) & \ddots &  & &  &  &   \\
            \vdots & a_{t-1}^{(t+1)}(1-s) & \ddots & a_{t+1}^{(t+1)}(1-s) & & &  & \\
             & \vdots & \ddots & a_{t}^{(t+1)}(1-s) & a_0^{(s)}(1-s) & & & &\\
             a_0^{(t+1)}(1-s) & & & \vdots & & \ddots & & \\
             & a_0^{(t+1)}(1-s) & & & & & & \\
             & & \ddots & & & &  &   \\
             & & & a_0^{(t+1)}(1-s) & & & &  a_0^{(s)}(1-s) 
        \end{pmatrix} \\
        &= (a_{t+1}^{(t+1)}(1-s))^s(a_0^{(s)}(1-s))^{t+1} \\
        &= (-1)^{s(t+1)}\frac{(t-s+1)^s(t-s + 2)^s\dots (2t -s + 1)^s}{[(t+1)!]^s} \\
        &\neq 0.
        \end{align}
        The third line above follows from the fact that the relevant matrix is lower triangular.
        Since the evaluation of $R_{t+1, s}$ at $1-s$ is nonzero, $R_{t+1, s}$ cannot be the zero polynomial. 
        
        Consequently, $R_{t+1, s}$ has at most finitely many roots. It then follows from Lemma~\ref{lemma:common-factor-property} that there at most finitely many values of $d\coloneqq \dim \h$ for which $P_{t+1}$ and $P_s$ may share a common root. By~\eqref{eq:m-1-zonal-spher-func}, it follows there are at most finitely many values of $d\coloneqq \dim \h$ for which there exists a unitary $u\in U(\h)$ such that $Z_{t+1}(u) = 0 = Z_s(u)$. Since $1\leq s \leq t$ was arbitrary, the first claim of the lemma holds.        
        
        Finally, we compute $R_{p, q}$ in terms of $d$ using computer algebra software and summarize the resultants in Table \ref{table:resultants} for $(p,q) \in \mathcal{C} \coloneqq \{(a, b) \in \mathbb{Z}^2: 1\leq a \leq 4, 1\leq b < a\}$. It may also be verified using computer algebra software that for $(p,q) \in \mathcal{C}$, the resultant $R_{p,q}$ has no integer roots in $[3, \infty)$.
        Therefore, $P_p$ and $P_q$ have no common roots for $p,q \in \mathcal{C}$ whenever $d \coloneqq \dim \h > 2$. The second claim of the lemma then follows from~\eqref{eq:m-1-zonal-spher-func}.
        
        \begin{table}[H]
        \renewcommand{\arraystretch}{1.5}
        \centering
        \begin{tabular}{||c | c||}
         \hline
         $(p,q)$ & $R_{p,q}$  \\ [0.5ex] 
         \hline\hline
         (2, 1) & $\frac{1}{2} (2 - d - d^2)$  \\ 
         \hline
         (3, 2) & $\frac{1}{288} (512 d^2 - 352 d^4 - 144 d^5 - 16 d^6)$  \\
         \hline
         (3, 1) & $\frac{1}{6}(-24 + 28 d - 4 d^3)$  \\
         \hline
         (4, 3) & $\frac{1}{17915904}(-90699264 d^2 - 309889152 d^3 - 331304256 d^4 + 
  10287648 d^5 + 299006640 d^6$ \\
  &$+ 269321760 d^7 + 118587888 d^8 + 
  29929824 d^9 + 4397328 d^{10} + 349920 d^{11} + 11664 d^{12})$ \\
         \hline
         (4, 2) & $\frac{1}{9216}(78400 d^2 - 184640 d^3 + 27536 d^4 + 71680 d^5 + 9632 d^6 - 
 2240 d^7 - 368 d^8)$  \\ 
         \hline
         (4, 1) & $\frac{1}{24}(360 - 618 d + 231 d^2 + 42 d^3 - 15 d^4)$  \\ 
         \hline

        \end{tabular}
        \caption{Resultants $R_{p,q}$ between polynomials $P_p$ and $P_q$ given by~\eqref{eq:P-t-poly-def}.}
        \label{table:resultants}
        \end{table}
        
    \end{proof}

\section{Proofs of Theorems~\ref{thm:main-thm-2} and~\ref{thm:main-thm-3}}\label{sec:group-thms}

As before, let $\h$ be a finite-dimensional Hilbert space, let $S\leq U(\h)$ be closed, and assume the notation of Section~\ref{sec:defs-and-prelims}, in particular the decompositions~\eqref{eq:decomposition}-\eqref{eq:u-h-s-action}. We will also use the commutant machinery introduced in~\eqref{eq:comm-def}-\eqref{eq:comm-equality}.

We now state the key observation needed for the proof of Theorems~\ref{thm:main-thm-2} and~\ref{thm:main-thm-3}, which has been noted in a slightly different context in~\cite{liu2024unitarydesignsrandomsymmetric}.
\begin{lemma}\label{thm:ordinary-designs-thm}
    Let $t>0$ and suppose $H$ is a finite subgroup of $U(\h)_S$ such that $(H, \mathrm{unif})$ is an $S$-symmetric $t$-design. Then for all $1\leq \lambda\leq n$, the uniformly weighted $U(\h)$-subsets $( \sigma_\lambda(H), \mathrm{unif})$ are all $t$-designs in $U(W_\lambda)$.
\end{lemma}

\begin{proof}
    Suppose $(H, \mathrm{unif})$ is an $S$-symmetric $t$-design. Then we have $\comm(\{h^{\otimes t}: h\in H\}) = \comm(\{u^{\otimes t}: u\in U(\mathcal{H})_S\})$. For all $h\in H$, we have 
    \begin{equation}\label{eq:g-t-decomp-form}
        h^{\otimes t} \cong \left(\bigoplus_{\lambda = 1}^n I_{V_\lambda} \otimes \sigma_\lambda(h) \right)^{\otimes t} \cong \bigoplus_{\mathbf{p} \in [n]^t} I_{V_{\mathbf{p}}} \otimes \sigma_{\mathbf{p}}(h),
    \end{equation}
    where $V_{\mathbf{p}} \cong V_{p_1} \otimes \dots \otimes V_{p_t}$ and $\sigma_{\mathbf{p}}(h) \cong \sigma_{p_1}(h) \otimes \dots \otimes \sigma_{p_t}(h)$. 
    Suppose $X$ lies in $\comm(\{\sigma_\lambda(h)^{\otimes t}: h\in H\}).$
    It follows from~\eqref{eq:g-t-decomp-form} that
    \begin{equation}
        \tilde{X} \coloneqq I_{V_{(\lambda, \dots, \lambda})}\otimes X \in \comm(\{h^{\otimes t}: h\in H\}) = \comm(\{u^{\otimes t}: u\in U(\mathcal{H})_S\}).
    \end{equation}
    Let $u \in U(W_\lambda)$ be arbitrary. We have that
    \begin{equation}
         \tilde{u} \coloneqq \left( [I_{V_{\lambda}} \otimes u] \oplus \bigoplus_{\stackrel{\lambda' = 1}{\lambda' \neq \lambda}}^n I_{V_{\lambda'}} \otimes I_{W_{\lambda'}} \right)^{\otimes t} \cong  (I_{V_{(\lambda, \dots, \lambda)}} \otimes u^{\otimes t}) \oplus  \bigoplus_{\stackrel{\mathbf{p}\in [n]^t}{\mathbf{p} \neq (\lambda, \dots, \lambda)}} I_{V_{\mathbf{p}}} \otimes O_{\mathbf{p}},
    \end{equation}

    where $O_{\mathbf{p}} \in U(W_{p_1} \otimes \dots \otimes W_{p_n})$ is an operator whose exact form is irrelevant for us. Since $\tilde{u}$ is the $t$-th tensor power of an element of $U(\mathcal{H})_S$, it commutes with $\tilde{X}$. Since $\tilde{X}$ and $\tilde{u}$ commute, their blocks corresponding to the index $(\lambda, \dots, \lambda)$ must commute. Hence, $X \in \comm(\{u^{\otimes t}: u \in U(W_\lambda)\})$. Therefore, 
    \begin{equation}
        \comm(\{\sigma_\lambda(h)^{\otimes t}: h\in H\}) \subseteq \comm(\{u^{\otimes t}: u\in U(W_\lambda)\}).
    \end{equation}
    Since $\sigma_\lambda(H) \leq U(W_\lambda)$, the set inclusion is in fact an equality. 
\end{proof}

A technical fact needed for the proofs of Theorems~\ref{thm:main-thm-2} and~\ref{thm:main-thm-3} is provided by the following lemma, which is essentially a restatement of Theorem 11 from Appendix F in~\cite{Mitsuhashi_2023} adapted to the language of this paper. 

\begin{lemma}\label{lemma:weighted-t-des}
    Let $K\leq U(\h)$ be a compact group and let $G\leq K$ be a finite group. If $p\colon  G \to [0,1]$ is a pmf such that $(G, p)$ is a $t$-design in $K$, then $(G, \mathrm{unif})$ is a $t$-design in $K$ as well. 
\end{lemma}

\begin{proof}
    Observe that 
    \begin{align}
        \sum_{g\in G} \frac{1}{|G|} g^{\otimes t} \otimes \overline{g}^{\otimes t} &=
        \sum_{g'\in G} p(g') \sum_{g\in G} \frac{1}{|G|} (g'g)^{\otimes t} \otimes (\overline{g'g})^{\otimes t} \\
        &=
        \sum_{g\in G} \frac{1}{|G|}\left(\sum_{g'\in G} p(g') (g')^{\otimes t} \otimes (\overline{g'})^{\otimes t}\right)(g^{\otimes t} \otimes \overline{g}^{\otimes t}) \\
        &= 
        \sum_{g\in G} \frac{1}{|G|} \left( \int_{K} dk \ k^{\otimes t} \otimes \overline{k}^{\otimes t}\right) (g^{\otimes t} \otimes \overline{g}^{\otimes t}) \\
        &=  \sum_{g\in G} \frac{1}{|G|} \left( \int_K dk \ (kg)^{\otimes t} \otimes (\overline{kg})^{\otimes t} \right) \\
        &= \int_K dk \ k^{\otimes t} \otimes \overline{k}^{\otimes t},
    \end{align}
    where the integration is with respect to the normalized Haar measure over $G$. The first and fifth equalities follow from the invariance of the Haar measure. 
\end{proof}

Let us prove the following main result.
\begin{theorem}[Theorem~\ref{thm:main-thm-2} restated]\label{thm:main-thm-2-restated}
    The following hold:
    \begin{enumerate}
        \item If the natural representation of $S$ on $\h$ is multiplicity free, then for any $t>0$, there exists a finite group $H \leq U(\h)_S$ such that $t_{\max}^{(S)}(H, \mathrm{unif}) \geq t$. Consequently, for any non-trivial finite-dimensional Hilbert space $\h'$ and any $t>0$, there exists a symmetry group $S'\leq U(\h')$ and finite group $H\leq U(\h')_{S'}$ for which $t_{\max}^{(S')}(H, \mathrm{unif}) \geq t$ and $U(\h')_{S'} \neq \left\{e^{i\theta} I: \theta \in \mathbb{R} \right\}$.

        \item Let $G\leq U(\h)$ be finite. If $t_{\max}^{(S)}(G, p) > 5$ for some pmf $p$ over $G$, then every element of $U(\h)_{S}$ is diagonal in the irrep basis of $S$ and, hence, $U(\h)_{S}$ is abelian.

        \item For any finite-dimensional Hilbert space $\h'$, there exists a symmetry group $S'\leq U(\h')$ and finite group $H\leq U(\h')_{S'}$ for which $t_{\max}^{(S')}(H, \mathrm{unif}) = 5$ and $U(\h')_{S'}$ is non-abelian.
    \end{enumerate}
\end{theorem}
\begin{proof}
    To prove the first part of the theorem, note that Proposition~\ref{prop:symmetric-design-construction} implies that whenever the natural representation of $S$ on $\h$ is multiplicity-free, so that $\dim W_\lambda = 1$ for $1\leq \lambda\leq n$, we have that
    \begin{equation}
        H \coloneqq \left\{ \bigoplus_{\lambda = 1}^n \exp\left(\frac{2\pi i j_\lambda}{t+1}\right) I_{V_\lambda} \otimes I_{W_\lambda}: 0\leq j_\lambda \leq t\right\}
    \end{equation}
    is a finite subgroup of $U(\h)_S$ for which $(H, \mathrm{unif})$ is an $S$-symmetric $t$-design. To see that we can indeed find a closed symmetry group in $U(\h')$ whose natural representation on $\h'$ is multiplicity-free for a Hilbert space $\h'$ of arbitary dimension, let $V_1, \dots, V_n$, $n>1$ be a sequence of orthogonal subspaces of $\h'$ for which $\h' = \bigoplus_{\lambda = 1}^n V_\lambda$, let $S_\lambda \leq U(V_\lambda)$ be a non-trivial closed group that acts irreducibly on $V_\lambda$, and define 
    \begin{equation}
        S' \coloneqq \left\{ \bigoplus_{\lambda = 1}^n s_\lambda: s_\lambda \in S_\lambda \right\};
    \end{equation}
    in this case, the subspaces $V_1, \dots, V_n$ are all inequivalent irreps of the group $S'$, so that Schur's lemma implies
    \begin{equation}
        U(\h')_{S'} = \left\{\bigoplus_{\lambda = 1}^n e^{i\theta_\lambda} I_{V_\lambda}: \theta_\lambda \in \mathbb{R}\right\} \neq \{e^{i\theta} I: \theta \in \mathbb{R}\}.
    \end{equation}

    Let us now prove the second part of the theorem. Suppose that $G\leq U(\h)$ is finite and that $t\coloneqq t_{\max}^{(S)}(G, p) > 5$ for some pmf $p$ over $G$. Since $(G_S, p_S)$ is a $6$-design in $U(\h)_S$, so is $(G_S, \mathrm{unif})$ by Lemma~\ref{lemma:weighted-t-des}. Consequently, $(\sigma_\lambda(G_S), \mathrm{unif})$ is a $6$-design in $U(W_\lambda)$ for all $\lambda \in [n]$ by Lemma~\ref{thm:ordinary-designs-thm}. By the classification of finite group designs in~\cite{Bannai}, it follows that $\dim W_\lambda = 1$ for all $\lambda \in [n]$. Therefore, we have 
    \begin{equation}
        U(\h)_S = \left\{ \bigoplus_{\lambda = 1}^n I_{V_\lambda} \otimes u: u\in U(W_\lambda)\right\} = \left\{ \bigoplus_{\lambda = 1}^n e^{i\theta_\lambda} I_{V_\lambda} \otimes I_{W_\lambda}: \theta_\lambda \in \mathbb{R}\right\},
    \end{equation}
    where the first equality follows from Schur's Lemma and the second equality follows directly from the fact that $\dim W_\lambda = 1, \lambda \in [n]$.
    
    Finally, we prove the third part of the theorem. Let $d>0$ be arbitrary and let $\h'$ be a Hilbert space of dimension $d$. Let us assume that $d$ is odd, so that $d = 2k + 1$ for some $k > 0$; our argument requires minimal modifications when $d$ is even. Let $W_1, \dots, W_k$ be a sequence of orthogonal 2-dimensional subspaces of $\h'$ and let $W_{k+1}$ be a 1-dimensional subspace of $\h'$ for which $\h' = \bigoplus_{\lambda =1}^{k+1} W_\lambda$. By~\cite{Bannai, Gross2007}, for each $\lambda \in [n]$, there exists a finite subgroup $H_\lambda \leq U(W_\lambda)$ for which $(H_\lambda, \mathrm{unif})$ is a $5$-design in $U(W_\lambda)$. (For $1\leq \lambda \leq k$, this group $H_\lambda$ is isomorphic to $\SL(2, 5)$, which is defined as the multiplicative group of $2\times 2$ determinant-one matrices with elements from the finite field of order 5.) Then let 
    \begin{equation}
        S' \coloneqq \left\{\bigoplus_{\lambda=1}^{k+1} e^{i\theta_\lambda} I_{W_\lambda}: \theta_\lambda \in \mathbb{R} \right\}.
    \end{equation}
    Let $C_\lambda \coloneqq \langle \exp(2\pi i/6) I_{W_\lambda}\rangle$ for $\lambda \in [n]$.
    Naturally, the set
    \begin{equation}
        H = \left\{ \bigoplus_{\lambda = 1}^n g_\lambda: g_\lambda \in C_\lambda H_\lambda\right\}
    \end{equation}
    is a subgroup of $U(\h')_{S'}$. Define also the pmf $p$ over $H$ by
    \begin{equation}
        p\left(\bigoplus_{\lambda = 1}^n g_\lambda\right) = \prod_{\lambda = 1}^n (\mu_\lambda \star \nu_\lambda)(g_\lambda), 
    \end{equation}
    where $\mu_\lambda$ and $\nu_\lambda$ are the uniform pmfs over $C_\lambda$ and $H_\lambda$. By Proposition~\ref{prop:symmetric-design-construction}, the weighted $U(\h')_{S'}$-subset $(H, p)$ forms a $5$-design in $U(\h')_{S'}$. By Lemma~\ref{lemma:weighted-t-des}, it follows that $(H, \mathrm{unif})$ is a 5-design in $U(\h')_{S'}$. Note that the subspaces $W_1, \dots, W_{k+1}$ are all inequivalent representations of $S'$ (though not necessarily irreducible), so that Schur's Lemma implies
    \begin{equation}
        U(\h')_{S'} = \left\{ \bigoplus_{\lambda = 1}^{k+1} u_\lambda: u_\lambda \in U(W_\lambda) \right\}.
    \end{equation}
    But since $\dim W_\lambda > 1$ for $1\leq \lambda \leq k$, the group $U(\h')_{S'}$ cannot be abelian. By the second part of the statement of the theorem, it follows that $(H, \mathrm{unif})$ cannot be a 6-design in $U(\h')_{S'}$, so that $t_{\max}^{(S')}(H, \mathrm{unif}) = 5$. 

    \end{proof}

    We now prove our next main result. Recall that $\SL(2, 5)$ denotes the multiplicative group of $2\times 2$ determinant-one matrices with elements from the finite field of order 5.

    \begin{theorem}[Theorem~\ref{thm:main-thm-3} restated]\label{thm:main-thm-3-restated}
        The following hold:
        \begin{enumerate}
            \item Let $G\leq U(\h)$ be finite. Let $q\geq 2, N\geq 4$. If $\h = (\mathbb{C}^2)^{\otimes N}$ or $\h = (\mathbb{C}^q)^{\otimes N}$ and 
        \begin{align}
            S = \left\{(e^{i\theta Z})^{\otimes N}: \theta \in \mathbb{R}\right\} \cong \mathrm{U}(1)
        \end{align}
        or 
        \begin{align}
            S = \left\{u^{\otimes N}: u \in \mathrm{SU}(q)\right\} \cong \mathrm{SU}(q),
        \end{align}
        then $t_{\max}^{(S)}(G, p) \leq 2$ for any pmf $p\colon  G\to [0,1]$.
    
        \item  Let $G\leq U(\h)$ be finite and assume that $S\leq U(\h)$ is such that $|S| < \dim \h/2$. Then $t_{\max}^{(S)}(G, p) \leq 3$ for any pmf $p\colon  G\to [0,1]$.
    
        \item Let $G\leq U(\h)_S$ be finite. If $4\leq t_{\max}^{(S)}(G, p) \leq 5$ for some pmf $p$ over $G$, then 
        \begin{equation}
            G \leq \left\{ \bigoplus_{\lambda = 1}^n I_{V_\lambda} \otimes s_\lambda: s_\lambda \in \left\langle \exp\left(\frac{2\pi i}{k_\lambda} \right) I_{W_\lambda} \right\rangle \rho_\lambda(\SL(2, 5)) \right\}
        \end{equation}
        for some integers $k_1,\dots, k_n$ and irreducible representations $\rho_\lambda: \SL(2, 5) \to U(W_\lambda), \lambda \in [n]$ of either degree one or degree two; furthermore, for every $s_\lambda \in \langle \exp(2\pi i/k_\lambda) I_{W_\lambda}\rangle \rho_\lambda(\SL(2, 5))$, there exists an element $g\in G$ such that $\sigma_\lambda(g) = s_\lambda$.
        \end{enumerate}
    \end{theorem}

    \begin{proof}
        Let us begin by proving the first statement of the theorem. For the sake of contradiction, suppose that there exists a pmf $p\colon  G \to [0,1]$ for which $G_S$ is an $S$-symmetric 3-design under $p_S$. By Lemma~\ref{lemma:weighted-t-des}, it follows that $(G_S, \mathrm{unif})$ is a 3-design in $U(\h)_S$. By the Schur-Weyl duality, the Hilbert space $\h$ decomposes into $S$-irreps as
        \begin{equation}
            \h \cong \bigoplus_{\lambda \vdash_d n} V_\lambda \otimes W_\lambda, 
        \end{equation}
        where the direct sum runs over all Young diagrams consisting of $N$ boxes and at most $d$ rows, the $V_\lambda$ denote inequivalent $S$-irreps, and their corresponding multiplicity spaces $W_\lambda$ are $\mathfrak{S}_N$-irreps. 
    
        \textbf{Case 1: $N \notin \{7, 13, 19\}.$}
         By the hook length formula, we have
        \begin{align}
            \dim W_{(N-1, 1)} &= N-1 \\
            \dim W_{(N-2, 2)} &= \frac{1}{2}N(N-3). 
        \end{align}
        The weighted subsets $(\sigma_\lambda(G_S), \mathrm{unif})$ must be 3-designs in $U(W_\lambda)$ by Theorem~\ref{thm:ordinary-designs-thm}. Since the sets $\sigma_\lambda(G_S)$ are all finite subgroups of $U(W_\lambda)$, it follows from Corollary 2 and Theorem 4 of~\cite{Bannai} that each $W_\lambda$ must have dimension equal to either 1, 2, 6, 12, 18 or $2^a$ for some $a> 2$. We cannot have $\dim W_{(N-1, 1)} = N-1$ equal to 1, 2, 6, 12, or 18 by hypothesis, so $\dim W_{(N-1, 1)} = 2^a$ for $a>2$. Hence, $\dim W_{(N-2, 2)} = \frac{1}{2}(2^a+1)(2^a - 2)$, which is neither a power of two nor equal to 1, 2, 6, 12, or 18, yielding a contradiction.
    
        \textbf{Case 2: $N \in \{ 7, 13, 19\}$}. In this case, we must have $\dim W_{(N-2, 2)} = N(N-3)/2 \in \{14, 65, 152\}$. Hence, Theorem 4 of~\cite{Bannai} leads to a contradiction again.

        Since each case results in a contradiction, a pmf $p\colon  G\to [0,1]$ such that $(G_S, p_S)$ is a 3-design in $U(\h)_S$ cannot exist. Hence, $t_{\max}^{(S)}(G, p) \leq 2$ for any pmf $p\colon  G\to [0,1]$.

        Next, we prove the second statement of the theorem. Under the generic finite symmetry group $S$, the space $\h$ decomposes as before like
        \begin{equation}
            \h \cong \bigoplus_{\lambda = 1}^n V_\lambda \otimes W_\lambda,
        \end{equation}
        where the $V_\lambda$ are the inequivalent $S$-irreps and the $W_\lambda$ are their corresponding multiplicity spaces. Suppose towards contradiction that there exists a pmf $p\colon  G \to [0,1]$ for which $G_S$ is an $S$-symmetric 4-design under $p_S$. By Lemma~\ref{lemma:weighted-t-des}, it again follows that $(G_S, \mathrm{unif})$ is a 4-design in $U(\h)_S$, so that Lemma~\ref{thm:ordinary-designs-thm} implies the $(\sigma_\lambda(G_S), \mathrm{unif})$ are all 4-designs in $U(W_\lambda)$. Corollary 2 of~\cite{Bannai} then implies that $1\leq \dim W_\lambda \leq 2$ for all $\lambda \in [n]$. Consequently, 
        \begin{equation}\label{eq:irrep-deg-inequality}
            \dim \h = \sum_{\lambda=1}^n \dim V_\lambda \dim W_\lambda \leq 2\sum_{\lambda = 1}^n \dim V_\lambda \leq 2\sum_{\lambda = 1}^n (\dim V_\lambda)^2 \leq 2 |S|,
        \end{equation}
        where the last inequality follows from the general representation-theoretic fact that the order of a finite group is equal to the sum of the squares of the degrees of all of its inequivalent irreps. But~\eqref{eq:irrep-deg-inequality} contradicts our assumption that $|S| < \dim \h/2$. Therefore, $t_{\max}^{(S)}(G, p)\leq 3$ for any pmf $p\colon  G \to [0,1]$. 

        The third statement of the theorem is obtained by again using Lemmas~\ref{lemma:weighted-t-des} and~\ref{thm:ordinary-designs-thm} to argue that $(\sigma_\lambda(G), \mathrm{unif})$ is either a 4-design or 5-design in $U(W_\lambda)$, at which point Corollary 2 of Bannai~\cite{Bannai} implies that either \textit{(i)} $\dim W_\lambda = 1$ and $\sigma_\lambda(G)$ is a cyclic group of phases or \textit{(ii)} $\dim W_\lambda = 2$ and $\sigma_\lambda(G) = Z(\sigma_\lambda(G))H_\lambda$, where $Z(\sigma_\lambda(G))$ denotes the center of $\sigma_\lambda(G)$ and $H_\lambda \leq U(W_\lambda)$ is a group isomorphic to $\SL(2, 5)$. Because $(\sigma_\lambda(G), \mathrm{unif})$ is at least a 1-design in $U(W_\lambda)$, the group $\sigma_\lambda(G)$ must act irreducibly on $W_\lambda$ (see again Proposition 11 in~\cite{zhou2025}, Section VA in~\cite{Gross2007}, or Theorem 1 in~\cite{kaposi2026}). Therefore, in case \textit{(ii)}, the center of $\sigma_\lambda(G)$ must be a cyclic group of phases by Schur's Lemma and $H_\lambda$ must be the image of an irreducible representation of $\SL(2, 5)$. The third statement of the theorem then follows from the definition of the map $\sigma_\lambda$ and the fact that the only degree one irrep of $\SL(2, 5)$ is the trivial representation. 
    \end{proof}

\section{Proof of Theorem~\ref{thm:main-thm-4}}\label{sec:proof-of-grp-des}

We first establish a technical lemma. 
\begin{lemma}\label{lemma:supergroup-design}
    Let $K\leq U(\h)$ be a compact group and let $t>0$. Let $H\leq G\leq K$ be finite subgroups. If $(H, \mathrm{unif})$ is a $t$-design in $K$, then so is $(G, \mathrm{unif})$. 
\end{lemma}

\begin{proof}
    Suppose that $(H, \mathrm{unif})$ is a $t$-design in $K$. Let $G/H$ denote the collection of left cosets of $H$. Note that $G/H$ is well-defined even though it is not necessarily a group. Then
    \begin{align}
        \frac{1}{|G|} \sum_{g\in G} g^{\otimes t} \otimes (\overline{g})^{\otimes t} &= 
        \frac{1}{|G/H||H|} \sum_{gH \in G/H} \sum_{h\in H} (gh)^{\otimes t} \otimes (\overline{gh})^{\otimes t} \\
        &=
        \frac{1}{|G/H|} \sum_{gH \in G/H} g^{\otimes t} \otimes \overline{g}^{\otimes t} \left(\frac{1}{|H|} \sum_{h\in H} h^{\otimes t} \otimes \overline{h}^{\otimes t}\right) \\
        &= 
        \frac{1}{|G/H|} \sum_{gH\in G/H} (g^{\otimes t} \otimes \overline{g}^{\otimes t}) \int_K dk \ k^{\otimes t} \otimes \overline{k}^{\otimes t} \\
        &=
        \frac{1}{|G/H|} \sum_{gH \in G/H} \int_K dk \ (gk)^{\otimes t} \otimes (\overline{gk})^{\otimes t} \\
        &= \int_K dk \ k^{\otimes t} \otimes \overline{k}^{\otimes t}.
    \end{align}
    The first equality follows from the fact that $G/H$ partitions $G$. The last equality follows the left invariance of the Haar measure over $K$. 
\end{proof}

We will use the commutant machinery introduced in~\eqref{eq:comm-def}-\eqref{eq:comm-equality} in our proof of Theorem~\ref{thm:main-thm-4}

\begin{theorem}[Theorem~\ref{thm:main-thm-4} restated]\label{thm:main-thm-4-restated}
    Let $t>1$. Then for any non-trivial finite-dimensional Hilbert space $\h$, there exists a finite group $G\leq U(\h)$ and symmetry group $S\leq U(\h)$ such that $U(\h)_S \neq \{e^{i\theta} I: \theta \in \mathbb{R} \}$ and
    \begin{equation}\label{eq:group-t-max-statement-restated}
        t_{\max}^{(S)}(G, \mathrm{unif})\geq t > t_{\max}(G, \mathrm{unif}) = 1.
    \end{equation}
    Moreover, if $\dim \h$ is even and $t\leq 5$, then there exists a finite group $G\leq U(\h)$ and symmetry group $S\leq U(\h)$ such that~\eqref{eq:group-t-max-statement-restated} holds and $G_S$ is non-abelian.
\end{theorem}

\begin{proof}
    Let $\h$ be an arbitrary non-trivial finite-dimensional Hilbert space. Let us write $\dim \h = nK$, where $n\geq 1$ and $K$ is a positive integer such for any $K$-dimensional Hilbert space $\h_K$, there exists a finite subgroup $H\leq U(\h_K)$ for which $(H, \mathrm{unif})$ is a $t$-design. Note that we may always take $K = 1$, because there exists a uniformly weighted group $t$-design over any one-dimensional Hilbert space. Furthermore, if $\dim \h$ is even and $t\leq 5$, we may take $K = 2$ because there exists a (non-abelian) uniformly weighted group $5$-design over any two-dimensional Hilbert space~\cite{Bannai, Gross2007}.
    
    Let $W_1, \dots, W_n$ be a collection of pairwise orthogonal $K$-dimensional subspaces of $\h$ for which $\h = \bigoplus_{\lambda = 1}^n W_\lambda$. For $\lambda\in [n]$, let $\{|w_i^{(\lambda)}\rangle\}_{i\in [K]}$ be an orthonormal basis for $W_\lambda$, so that $\{|w_i^{(\lambda)}\rangle\}_{i\in [K], \lambda\in [n]}$ forms an orthonormal basis for $\h$. For $\lambda,\lambda' \in [n]$, define the isomorphism $\Pi_{\lambda \to \lambda'}: W_\lambda \to W_\lambda'$ by
        \begin{equation}
            \Pi_{\lambda\to\lambda'} = \sum_{i=1}^K |w_i^{(\lambda')}\rangle\langle w_i^{(\lambda)}|.
        \end{equation}
        Let $(G_1, \mathrm{unif})$ be a unitary $t$-design in $U(W_1)$, which must exist because $\dim W_1 = K$. For $\lambda \in [n]$, define $G_\lambda = \Pi_{1\to \lambda} G_1 \Pi_{\lambda \to 1}$. With the understanding that $\mathfrak{S}_n$ denotes the symmetric group over $n$ letters, let us define
        \begin{equation}
            G = \left\{\sum_{\lambda = 1}^n \exp\left(\frac{2\pi i}{t+1} j_\lambda\right) \Pi_{\lambda \to \sigma(\lambda)}g_\lambda: 0 \leq j_\lambda \leq t, \sigma \in \mathfrak{S}_n, g_\lambda \in G_\lambda\right\}.
        \end{equation}
        Define also
        \begin{equation}
            S = \left\{\bigoplus_{\lambda = 1}^n e^{i\theta_\lambda} I_{W_\lambda}: \theta_\lambda \in \mathbb{R} \right\}
        \end{equation}       
        and note that $S$ is a subgroup of $U(\h)$.
        We will show that $G$ is a subgroup of $U(\h)$ for which $t_{\max}(G, \mathrm{unif})= 1$ while $t_{\max}^{(S)}(G, \mathrm{unif})=t$. 

        First, note that for $a,b,c \in [n]$, we have \textit{(i)} $\Pi_{a\to a} = \Id_{W_a}$, \textit{(ii)} $\Pi_{a\to b}^\dagger = \Pi_{b\to a}$, \textit{(iii)} $\Pi_{b\to c} \Pi_{a\to b} = \Pi_{a\to c}$. With these identities, let us first verify that the elements of $G$ are unitary operators. We have
        \begin{align}
            &\left(\sum_{\lambda = 1}^n \exp \left( \frac{2\pi i}{t+1} j_\lambda \right) \Pi_{\lambda \to \sigma(\lambda)}g_\lambda \right)\left(\sum_{\lambda' = 1}^n\exp\left(\frac{2\pi i}{t+1} j_{\lambda'} \right)\Pi_{\lambda' \to \sigma(\lambda')} g_{\lambda'} \right)^\dagger \\ &= 
            \sum_{\lambda, \lambda' = 1}^n \exp\left(\frac{2\pi i}{t+1} (j_\lambda - j_{\lambda'}) \right) \Pi_{\lambda \to \sigma(\lambda)}g_\lambda g_{\lambda'}^\dagger \Pi_{\lambda' \to \sigma(\lambda')}^\dagger \\
            &= \sum_{\lambda = 1}^n \Pi_{\lambda \to \sigma(\lambda)}g_\lambda g_{\lambda}^\dagger \Pi_{\lambda \to \sigma(\lambda)}^\dagger \\
            &= \sum_{\lambda = 1}^n \Pi_{\lambda \to \sigma(\lambda)}\Pi_{\sigma(\lambda) \to \lambda} \\
            &= \sum_{\lambda = 1}^n I_{W_\lambda}\\
            &= I_\h,
        \end{align}
        where the third line follows from the fact that $g_\lambda$ and $g_{\lambda'}$ are supported on orthogonal subspaces if $\lambda \neq \lambda'$.

        Next, let us verify that $G$ is indeed a subgroup of $U(\h)$. Fix $\rho, \sigma\in \mathfrak{S}_n$ to be arbitrary permutations and for all $\lambda \in [n]$, let $g_\lambda, h_\lambda \in G_\lambda$ and $0\leq j_\lambda, k_\lambda \leq t$ be arbitrary. Let us write $g_\lambda = \Pi_{1\to \lambda} g^{(\lambda)} \Pi_{\lambda \to 1}$ and $h_\lambda = \Pi_{1\to \lambda} h^{(\lambda)} \Pi_{\lambda \to 1}$, where $g^{(\lambda)}, h^{(\lambda)} \in G_1$. Then
        \begin{align}
            &\left(\sum_{\lambda = 1}^n \exp\left(\frac{2\pi i}{t+1} j_\lambda \right) \Pi_{\lambda \to \sigma(\lambda)} g_\lambda \right)\left(\sum_{\lambda' = 1}^n \exp\left(\frac{2\pi i}{t+1} k_{\lambda'} \right) \Pi_{\lambda' \to \rho(\lambda')}h_{\lambda'} \right)^{\dagger} \\
            &= \sum_{\lambda, \lambda' = 1}^n \exp\left(\frac{2\pi i}{t+1} (j_\lambda - k_{\lambda'}) \right) \Pi_{\lambda \to \sigma(\lambda)} g_\lambda h_{\lambda'}^\dagger \Pi_{\lambda' \to \rho(\lambda')}^\dagger
            \\
            &= 
            \sum_{\lambda = 1}^n \exp\left( \frac{2\pi i}{t+1}(j_\lambda - k_{\lambda}) \right) \Pi_{\lambda \to \sigma(\lambda)} g_\lambda h_{\lambda}^\dagger \Pi_{\lambda \to \rho(\lambda)}^\dagger \\
            &= \sum_{\lambda = 1}^n \exp\left(\frac{2\pi i}{t+1}(j_\lambda - k_\lambda) \right) 
            \Pi_{\lambda \to \sigma(\lambda)}\Pi_{1\to \lambda}g^{(\lambda)}\Pi_{\lambda \to 1} \Pi_{1\to \lambda} h^{(\lambda)\dagger}\Pi_{\lambda \to 1}\Pi_{\rho(\lambda) \to \lambda} \\
            &= \sum_{\lambda = 1}^n \exp\left(\frac{2\pi i}{t+1}(j_\lambda - k_\lambda) \right)
            \Pi_{1\to \sigma(\lambda)} g^{(\lambda)}h^{(\lambda)}\Pi_{\rho(\lambda) \to 1} \\
            &= \sum_{\lambda = 1}^n \exp\left(\frac{2\pi i}{t+1}(j_\lambda - k_\lambda) \right)
            \Pi_{1\to \sigma(\lambda)} 
            \Pi_{\rho(\lambda) \to 1} \Pi_{1\to \rho(\lambda)} g^{(\lambda)}h^{(\lambda)}\Pi_{\rho(\lambda) \to 1} \\
            &= \sum_{\lambda = 1}^n \exp\left(\frac{2\pi i}{t+1}(j_\lambda - k_\lambda) \right)
            \Pi_{\rho(\lambda) \to \sigma(\lambda)} \Pi_{1\to \rho(\lambda)} g^{(\lambda)}h^{(\lambda)}\Pi_{\rho(\lambda) \to 1} \\
            &= 
            \sum_{\mu = 1}^n \exp\left( \frac{2\pi i}{t+1}(j_{\rho^{-1}(\mu)} - k_{\rho^{-1}(\mu)}) \right) \Pi_{\mu\to (\sigma\rho^{-1})(\mu)}\Pi_{1\to \mu} g^{(\rho^{-1}(\mu))}h^{(\rho^{-1}(\mu))}\Pi_{\mu \to 1} 
        \end{align}
        The last equality follows from the variable change $\mu = \rho(\lambda)$. Since $\Pi_{1\to \mu} g^{(\rho^{-1}(\mu))}h^{(\rho^{-1}(\mu))} \Pi_{\mu \to 1}$ lies in $G_\mu$ by its definition, we conclude that the operator in the last line is contained in $G$. Therefore, $G$ is a subgroup of $U(\h)$.

        Observe that since $G_1$ (resp. $U(W_1)$) is isomorphic to $G_\lambda$ (resp. $U(W_\lambda)$) via $g\mapsto \Pi_{1\to \lambda} g \Pi_{\lambda \to 1}$, we have that $(G_\lambda, \mathrm{unif})$ is a $t$-design in $U(W_\lambda)$. Then note that
        \begin{align}
            H \coloneqq \left\{ \bigoplus_{\lambda = 1}^n \exp\left( \frac{2\pi i}{t+1}j_\lambda \right)g_\lambda: 0\leq j_\lambda \leq t, g_\lambda \in G_\lambda \right\} \leq G
        \end{align}
        by construction and also that $[h, s]=0$ for all $h\in H, s\in S$. Consider the pmf $p: H \to [0,1]$ defined by
        \begin{equation}
            p\left(\bigoplus_{\lambda = 1}^n \exp\left(\frac{2\pi i}{t+1} j_\lambda \right) g_\lambda \right) \coloneqq \prod_{\lambda = 1}^n (\mu_\lambda \star \nu_\lambda)\left( \exp\left(\frac{2\pi i}{t+1} j_\lambda \right) g_\lambda \right), 
        \end{equation}
        where $\mu_\lambda$ and $\nu_\lambda$ are the uniform pmfs over $\left\langle \exp\left(2\pi i/(t+1) j_\lambda\right) I_{W_\lambda}\right\rangle$ and $G_\lambda$, respectively. It follows from Proposition~\ref{prop:symmetric-design-construction} of Section~\ref{sec:symmetry-boosted-design-proofs} that $(H, p)$ is a $t$-design in $U(\h)_S$. Then Lemma~\ref{lemma:weighted-t-des} from Section~\ref{sec:group-thms} implies that $(H, \mathrm{unif})$ is also a $t$-design in $U(\h)_S$. Finally, Lemma~\ref{lemma:supergroup-design} of the present section implies that $(G_S, \mathrm{unif})$ is a $t$-design in $U(\h)_S$. 

        Let us now show that $(G, \mathrm{unif})$ is a 1-design in $U(\h)$. Suppose that $L\in \comm(G)$. Naturally, $L$ commutes with every element of $H$. Note that the subspaces $W_\lambda$ are all invariant under $H$. Moreover, if $W_\lambda$ has a non-trivial proper subspace invariant under $H$, then then such a subspace is invariant under $G_\lambda$; since $G_\lambda \leq U(W_\lambda)$ with the uniform distribution is a $t$-design in $U(W_\lambda)$ for $t > 1$, the group $G_\lambda$ must act irreducibly on $W_\lambda$ by Lemma (see, for example, Proposition 11 in~\cite{zhou2025}, Section VA in~\cite{Gross2007}, or Theorem 1 in~\cite{kaposi2026}). Consequently, the subspaces $W_\lambda$ are all irreps of $H$. Moreover, the structure of $H$ implies that the $W_\lambda$ must be inequivalent irreps of $H$. Hence, Schur's Lemma implies that
        \begin{equation}
            L = \bigoplus_{\lambda = 1}^n c_\lambda I_{W_\lambda}
        \end{equation}
        for some $c_\lambda \in \mathbb{C}$. But since $L$ commutes with all of $G$, we must also have that
        \begin{align}
            \sum_{\lambda = 1}^n c_{\rho(\lambda)} \Pi_{\lambda \to \rho(\lambda)} = L\left(\sum_{\lambda = 1}^n \Pi_{\lambda \to \rho(\lambda)}\right) =
            \left(\sum_{\lambda = 1}^n \Pi_{\lambda \to \rho(\lambda)} \right)L =
            \sum_{\lambda = 1}^n c_\lambda \Pi_{\lambda \to \rho(\lambda)}
        \end{align}
        for all $\rho \in \mathfrak{S}_n$. It follows that $c_{\rho(\lambda)} = c_{\lambda}$ for all $\rho \in \mathfrak{S}_n$ and $\lambda \in [n]$. Consequently, $L$ is a scalar multiple of the identity, and so lies in $\comm(U(\h))$. Hence, $(G, p^{\mathrm{unif}})$ is a 1-design in $U(\h)$. 

        It remains to show that $(G, \mathrm{unif})$ is not a 2-design in $U(\h)$. For this purpose, consider the operator
        \begin{equation}
            L \coloneqq \bigoplus_{\lambda = 1}^n \mathbb{F}_{\lambda} \in \End(\h^{\otimes 2}), 
        \end{equation}
        where $\mathbb{F}_{\lambda}$ is the swap operator over the space $W_\lambda \otimes W_\lambda$. Let $g$ be an arbitrary element of $G$, so that we may write
        \begin{equation}
            g = \sum_{\lambda = 1}^n \exp\left(\frac{2\pi i}{t+1} j_\lambda\right) \Pi_{\lambda \to \sigma(\lambda)} g_\lambda.
        \end{equation}
        Then, 
        \begin{align}
            g^{\otimes 2} L &= \sum_{\mu, \lambda, \lambda' = 1}^n \exp\left(\frac{2\pi i}{t+1}(j_\lambda + j_{\lambda'})\right) \Pi_{\lambda \to \sigma(\lambda)} g_\lambda \otimes \Pi_{\lambda' \to \sigma(\lambda')} g_{\lambda'} \mathbb{F}_{\mu} \\
            &= \sum_{\lambda = 1}^n \exp\left(\frac{4\pi i}{t+1}j_\lambda\right) \Pi_{\lambda \to \sigma(\lambda)}g_\lambda \otimes \Pi_{\lambda \to \sigma(\lambda)} g_\lambda \mathbb{F}_\lambda \\
            &= \sum_{\lambda = 1} \exp\left(\frac{4\pi i}{t+1}j_\lambda\right) \mathbb{F}_{\sigma(\lambda)} \Pi_{\lambda \to \sigma(\lambda)} g_\lambda \otimes \Pi_{\lambda \to \sigma(\lambda)} g_\lambda \\
            &= \sum_{\mu, \lambda, \lambda = 1}^n \exp\left(\frac{2\pi i}{t+1}(j_\lambda + j_{\lambda'})\right) \mathbb{F}_{\mu} \Pi_{\lambda \to \sigma(\lambda)}g_\lambda \otimes \Pi_{\lambda' \to \sigma(\lambda')} g_{\lambda'} \\
            &= Lg^{\otimes 2}.
        \end{align}
        Here, the first and fourth equalities follow from the fact that $\mathbb{F}_{\mu}$ is supported exclusively on the subspace $W_\mu \otimes W_\mu \subset \h$, which is orthogonal to $W_\lambda \otimes W_{\lambda'}$ unless $\lambda = \lambda' = \mu$. Therefore, 
        $L \in \comm(\{g^{\otimes 2}: g\in G\})$. Note that because all subspaces of the form $W_{\lambda}\otimes W_{\lambda'}$ for $\lambda \neq \lambda'$ lie in the kernel of $L$, the operator $L$ cannot be written as a linear combination of the swap $\mathbb{F}$ and identity operators over the full Hilbert space $\h \otimes \h$, and hence does not lie in $\comm(\{u^{\otimes 2}: u\in U(\h)\})$. Therefore, $(G, \mathrm{unif})$ is not a 2-design in $U(\h)$.
        
    \end{proof}

\section{Discussion and future directions}

In our work, we construct finite unitary ensembles whose design strength is boosted under a broad family of symmetries. Such ensembles demonstrate an intriguing relationship between symmetry and randomness. However, our construction is limited by the fact that it only applies to symmetries satisfying a technical representation-theoretic condition, as inclusive as that condition may be. To establish that the existence of symmetry-boosted designs is truly a generic property of symmetry in the mathematical framework of designs, it is of interest to show that symmetry-boosted designs exist for all closed symmetry groups. An approach to proving such a result is to use the general condition for symmetry-boosted designs provided in Proposition~\ref{prop:symmetric-design-construction} and study the common zeros of the zonal orthogonal polynomials discussed in Section~\ref{sec:more-zonal-spher-func}. In a separate direction, our work provides an explicit construction of symmetry-boosted unitary designs but does not construct them in terms of concrete local quantum gates. Finding gate compilations is essential for symmetry-boosted designs to be experimentally useful in randomized protocols with symmetry constraints. 

We also discuss the existence and structure of symmetric designs with a group structure and derive various bounds on the symmetric design strength of general unitary ensembles whose operators form finite groups. In particular, we show that any such ensemble must have a symmetric design strength of at most two for global on-site $\mathrm{U}(1)$ and $\mathrm{SU}(2)$ symmetries. It seems plausible that with more combinatorial arguments, the proof of Theorem~\ref{thm:main-thm-3} in Section~\ref{sec:proof-of-grp-des} may be strengthened so that the upper bound on symmetric design strength of general finite group ensembles is indeed one for all $N$ sufficiently large, thus matching the symmetric design strength of the Clifford group under these symmetries~\cite{Mitsuhashi_2023}. Finally, it is interesting to consider whether the property that the symmetric design strength of the uniformly weighted Clifford group is strictly upper bounded by its unitary design strength~\cite{Mitsuhashi_2023} is a consequence of the detailed structure of the Clifford group, or whether this is merely a property of any finite group $t$-design, \textit{i.e.}, a uniformly weighted finite group that is a unitary $t$-design, for $t>1$. We emphasize that this property does not hold for general finite groups of unitaries with a unitary design strength of one, as we have indeed constructed explicit families of such finite group ensembles and symmetry groups in arbitrary dimension that have a symmetric design strength strictly greater than one in Theorem~\ref{thm:main-thm-4}. Nevertheless, it may be possible to generalize the bounding relationship between the symmetric design strength and unitary design strength of the Clifford group to arbitrary group 2-designs using more sophisticated group-theoretic arguments. 

\paragraph*{Acknowledgments.}  C.V. would like to thank Shivan Mittal, Bin Yan, and Benjamin Lovitz for helpful and encouraging discussions.
This work was supported by National Science Foundation Grant No.~2442410.
\printbibliography

\appendix
\section{Convolutions of finitely-supported functions over groups}\label{app:associativity}

Let $K$ be an arbitrary group. For any \textit{finite} $\mathcal{X}, \mathcal{Y} \subset K$ and functions $p\colon  \mathcal{X} \to \mathbb{C}$ and $q: \mathcal{Y} \to \mathbb{C}$, we may define their convolution $p\star q: \x\y \to \mathbb{R}$ by 
\begin{equation}
    (p\star q)(g) = \sum_{\stackrel{(x,y) \in \x\times \y:}{xy = g}} p(x)q(y).
\end{equation}
Here, we establish some elementary properties of convolutions as defined here. 

\begin{proposition}\label{prop:conv-properties}
    The following hold:
    \begin{enumerate}
        \item The convolution of functions over finite subsets of $K$ is associative.
        \item For any sequence of complex-valued functions $p_1: \x_1 \to \mathbb{C}$, \dots, $p_k: \x_k \to \mathbb{C}$ with $\x_1, \dots, \x_k \subset K$ finite, we have
        \begin{equation}
              (p_1\star \dots \star p_k)(g) = \sum_{\stackrel{(x_1, \dots, x_k) \in \x_1 \times \dots \times \x_k:}{x_1\dots x_k = g}} p_1(x_1) \dots p_k(x_k) 
        \end{equation}
        for all $g\in \x_1 \dots \x_k$.
       \item For any sequence of complex-valued functions $p_1: \x_1 \to \mathbb{C}$, \dots, $p_k: \x_k \to \mathbb{C}$ with $\x_1, \dots, \x_k \subset K$ finite and any function $F: \x_1 \dots \x_k \to \mathbb{C}$, we have
       \begin{equation}
              \sum_{g\in \x_1 \dots \x_k}  (p_1\star \dots \star p_k)(g) F(g) = \sum_{x_1 \in \x_1} \dots \sum_{x_k \in \x_k} \left( p_1(x_1) \dots p_k(x_k) \right) F(x_1 \dots x_k).
       \end{equation}
     \end{enumerate}
\end{proposition}

\begin{proof}
Let $p,q,r$ be complex-valued functions defined over finite subsets $\x,\y,\z \subset K$, respectively.
\begin{align}
        \left(p\star (q\star r)\right)(g) &= \sum_{\stackrel{(x, a) \in \x \times \y\z:}{xa = g}} p(x)(q\star r)(a) \\ \label{eq:conv-assoc-1}
        &= \sum_{\stackrel{(x, a) \in \x \times \y\z:}{xa = g}}\sum_{\stackrel{(y, z) \in \y \times \z:}{yz = a}} p(x)q(y)r(z) \\
        &= \sum_{\stackrel{(x,y, z) \in \x\times\y\times \z:}{xyz = g}} p(x)q(y)r(z) \\ \label{eq:conv-assoc-3}
        &= \sum_{\stackrel{(a, z) \in \x\y \times \z:}{az = g}} \sum_{\stackrel{(x, y) \in \x\times \y:}{xy = a}} p(x)q(y)r(z) \\
        &= \sum_{\stackrel{(a, z) \in \x\y \times \z:}{az = g}} (p\star q)(a) r(z) \\
        &= \left((p\star q)\star r\right)(g).
\end{align} 
Therefore, the convolution is associative, so that the notation $p_1 \star \dots \star p_k$ is unambiguous for any sequence of complex-valued functions $p_1, \dots, p_k$ over finite subsets of $K$. By combining the computations between~\eqref{eq:conv-assoc-1} and~\eqref{eq:conv-assoc-3} with an inductive argument, the second claim of the proposition follows. The third claim of the proposition immediately follows from the second. 
\end{proof}

\section{Constructing block diagonal designs}\label{app:block-diag-designs}
    We first quote a simple algebraic identity that has previously been used in a designs context~\cite{Iosue_2024, Iosue_2024_PRX, NAKATA_2013}. To provide the statement, we establish the notation that for any integer $m\geq 1$, tuple $\mathbf{p} \in [n]^t$ and $\pi \in S_t$, we set $\pi(\mathbf{p}) \coloneqq (p_{\pi(1)}, \dots, p_{\pi(t)})$. 
    \begin{lemma}\label{lemma:exponentials-app}
        Let $m, t\geq 1$. Then for any $\mathbf{p}, \mathbf{q} \in [m]^t$ we have
        \begin{multline}
        \label{eq:exp-identity}
        \frac{1}{(t+1)^m}\sum_{j_1 = 0}^t \dots \sum_{j_n = 0}^t \exp\left( \frac{2\pi i}{t+1}(j_{p_1} - j_{q_1} + j_{p_2} - j_{q_2} + \dots + j_{p_t} - j_{q_t})\right) \\
        =
        \begin{cases} 
        1, & \textnormal{if there exists } \pi \in S_t\textnormal{ s.t. } \mathbf{p} = \pi(\mathbf{q})\\
        0, & \textnormal{otherwise}. 
        \end{cases}
        \end{multline}
    \end{lemma}
    
  \begin{proposition}\label{prop:block-diag-designs}
        Let $m,t\geq 1$. Suppose that 
        \begin{equation}
            \h = \bigoplus_{\lambda = 1}^m P_\lambda \otimes Q_\lambda
        \end{equation}
        and let
        \begin{equation}
            K \coloneqq \left\{\bigoplus_{\lambda = 1}^m I_\lambda \otimes u_\lambda: u_\lambda \in U(Q_\lambda)\right\}. 
        \end{equation}
        For each $\lambda \in [n]$, let $C_\lambda \coloneqq \langle \exp(2\pi i/(t+1)) I_{Q_\lambda} \rangle$, let $\mu_\lambda$ denote the uniform distribution over $C_\lambda$, let $(H_\lambda, \nu_\lambda)$ be a $t$-design in $U(Q_\lambda)$, and let $\tilde{H}_\lambda \coloneqq C_\lambda H_\lambda$. Define
        \begin{equation}
            H \coloneqq \left\{\bigoplus_{\lambda = 1}^m I_\lambda \otimes g_\lambda: g_\lambda \in \tilde{H}_\lambda \right\} 
        \end{equation}
        and construct a probability distribution $p$ over $H$ by defining
        \begin{equation}
            p\left(\bigoplus_{\lambda = 1}^m I_{P_\lambda} \otimes g_\lambda \right) = \prod_{\lambda = 1}^n (\mu_\lambda \star \nu_\lambda)(g_\lambda). 
        \end{equation}
        Then the the weighted $K$-subset $(H, p)$ is a $t$-design in $K$.  
    \end{proposition}

    \begin{proof}
         We first recall that for any function $F$ over $\tilde{H}_\lambda$, we have
         \begin{equation}\label{eq:block-diagonal-des-sum-decomp}
             \sum_{g \in \tilde{H}_\lambda} (\mu_\lambda \star \nu_\lambda)(g)F(g) = \frac{1}{t+1}\sum_{j = 0}^t \sum_{h \in H_\lambda} \nu_\lambda(h) F\left(\exp(2\pi i j/(t+1))h\right)
         \end{equation}
         by the definition of the $\star$-product (see \Cref{app:associativity}).
         For $\mathbf{p} \in [n]^t$, let us write $S_t(\mathbf{p}) \coloneqq\{\pi(\mathbf{p}): \pi \in S_t\}$. We then compute: 
        \begin{align}
            &\sum_{h\in H} \nu(h)h^{\otimes t} \otimes (h^\dagger)^{\otimes t}\notag\\
            &= 
            \sum_{g_1 \in \tilde{H}_1} \dots \sum_{g_m \in \tilde{H}_m} \prod_{\lambda =1}^m (\mu_\lambda \star \nu_\lambda)(g_\lambda) \left(\bigoplus_{\lambda = 1}^m I_{P_\lambda} \otimes g_\lambda \right)^{\otimes t} \otimes \left(\bigoplus_{\lambda = 1}^m I_{P_\lambda} \otimes g_\lambda^\dagger \right)^{\otimes t} \\
            &\cong \sum_{g_1 \in \tilde{H}_1} \dots \sum_{g_m \in \tilde{H}_m} \prod_{\lambda = 1}^m (\mu_\lambda \star \nu_\lambda)(g_\lambda)
            \bigoplus_{\mathbf{p} \in [m]^t}  
            \bigoplus_{\mathbf{q} \in [m]^t} 
            \bigotimes_{i=1}^t 
            \left(I_{P_{p_i}} \otimes I_{Q_{q_i}} \otimes g_{p_i} \otimes g_{q_i}^\dagger\right) \\
            &=  
            \bigoplus_{\mathbf{p} \in [m]^t}  
            \bigoplus_{\mathbf{q} \in [m]^t} 
            \sum_{g_1 \in \tilde{H}_1} \dots \sum_{g_m \in \tilde{H}_m} 
            \prod_{\lambda = 1}^m (\mu_\lambda \star \nu_\lambda)(g_\lambda)
            \bigotimes_{i=1}^t 
            \left(I_{P_{p_i}} \otimes I_{Q_{q_i}} \otimes  g_{p_i} \otimes g_{q_i}^\dagger\right) \\
            &= \bigoplus_{\mathbf{p} \in [m]^t}  
            \bigoplus_{\mathbf{q} \in [m]^t}
            \frac{1}{(t+1)^m} \sum_{j_1=0}^t \sum_{h_1\in H_1} \nu_1(h_1) \dots \sum_{j_m = 0}^t \sum_{h_m\in H_m} \nu_m(h_m)
            \bigotimes_{i=1}^t 
            \left(I_{P_{p_i}} \otimes I_{Q_{q_i}} \otimes  h_{p_i} \otimes h_{q_i}^\dagger \right) \notag\\
            &\ \ \ \ \ \ \ \ \ \ \ \ \ \ \ \ \ \ \ \ \ \ \ \ \ \times \exp\left( \frac{2\pi i}{t+1}(j_{p_1} - j_{q_1} + j_{p_2} - j_{q_2} + \dots + j_{p_t} - j_{q_t})\right)  \label{eq:sums-inside-direct-sums}\\
            &= \bigoplus_{\mathbf{p} \in [m]^t} \bigoplus_{\mathbf{q} \in S_t(\mathbf{p})} \sum_{h_1\in H_1} \nu_1(h_1) \dots 
            \sum_{h_m\in H_m} \nu_m(h_m)
            \bigotimes_{i=1}^t 
            \left(I_{P_{p_i}} \otimes I_{Q_{q_i}} \otimes h_{p_i} \otimes h_{q_i}^\dagger\right)  \label{eq:apply-lemma}\\
            &= \bigoplus_{\mathbf{p} \in [m]^t} \bigoplus_{\mathbf{q} \in S_t(\mathbf{p})} 
            \int_{U(Q_1)} du_1 \dots
            \int_{U(Q_m)} du_m 
            \bigotimes_{i=1}^t 
            \left(I_{P_{p_i}} \otimes I_{Q_{q_i}} \otimes u_{p_i} \otimes u_{q_i}^\dagger\right) \label{eq:H_nu-design-property} \\
            &= \int_{U(Q_1)} du_1 \dots
            \int_{U(Q_m)} du_m 
            \bigoplus_{\mathbf{p} \in [m]^t} \bigoplus_{\mathbf{q} \in S_t(\mathbf{p})} \bigotimes_{i=1}^t \left(I_{P_{p_i}} \otimes I_{Q_{q_i}} \otimes u_{p_i} \otimes u_{q_i}^\dagger\right) \\
            &=\int_{U(Q_1)} du_1 \dots
            \int_{U(Q_m)} du_m 
            \bigoplus_{\mathbf{p} \in [m]^t} \bigoplus_{\mathbf{q} \in [m]^t} \bigotimes_{i=1}^t \left(I_{P_{p_i}} \otimes I_{Q_{q_i}} \otimes u_{p_i} \otimes u_{q_i}^\dagger\right) \label{eq:permutations-go-away}\\
            &\cong \int_{U(Q_1)} du_1 \dots
            \int_{U(Q_m)} du_m \left(\bigoplus_{\lambda = 1}^m I_{P_\lambda} \otimes u_\lambda \right)^{\otimes t} \otimes \left(\bigoplus_{\lambda = 1}^m I_{P_\lambda} \otimes u_\lambda^\dagger \right)^{\otimes t} \\
            &= \int_{K} du \ u^{\otimes t} \otimes (u^\dagger)^{\otimes t}. \label{eq:K-as-product-fubini}
        \end{align}
    The first equality follows from the assumed form of $H$. In line \eqref{eq:sums-inside-direct-sums} we bring the sums over $\tilde{H}_\lambda$ inside the direct sum and use~\eqref{eq:block-diagonal-des-sum-decomp} to decompose each expectation value over $\tilde{H}_\lambda$ into an expectation value over phases and an expectation value over $H_\lambda$. Line \eqref{eq:apply-lemma} follows from Lemma~\ref{lemma:exponentials-app}. Line \eqref{eq:H_nu-design-property} follows from the $t$-design property of each $(H_\lambda, \nu_\lambda)$, which applies here because the number of appearances of $g_\lambda$ in the tensor product is equal to the number of appearances of $g_\lambda^\dagger$ since $\mathbf{q}$ is a permutation of $\mathbf{p}$. Line \eqref{eq:permutations-go-away} follows from the fact that integrals of the form $\int_{U(Q_\lambda)} du \ u^{\otimes r} \otimes (u^\dagger)^{\otimes s}$ generically vanish when $r\neq s$. 
    Finally, in \label{eq:K-as-product-fubini} we use the fact that $K\cong U(Q_1) \times \dots \times U(Q_m)$ and Fubini's theorem for the Haar measure. 
\end{proof}

\end{document}